\documentclass[a4paper, 12pt]{article}
\pdfoutput=0
\usepackage[utf8]{inputenc}
\usepackage[T1]{fontenc}
\usepackage[english]{babel}
\usepackage{lmodern}
\usepackage{amsmath,amssymb,amsthm,bm,bbm,dsfont}
\usepackage{xcolor,graphicx,psfrag}
\usepackage{natbib}
\usepackage{hyperref}
\usepackage{subfig,floatrow}
\usepackage[font=footnotesize]{caption}
\usepackage{algorithm}
\usepackage{algpseudocode}
\usepackage{multibib}
\newcites{SM}{References for the Supplementary Material}

\def\PaperTitle{%
  Sequential design of multi-fidelity computer experiments: %
  maximizing the rate of stepwise uncertainty reduction}
\def\PaperKWA{%
  Multi-fidelity, %
  Computer experiments, %
  Sequential design of experiments}
\def\PaperKWB{%
  Gaussian process emulator,
  Meta-model, Surrogate model,
  Stochastic simulator}

\definecolor{Blue}{rgb}{0.1,0.1,0.5}
\hypersetup{bookmarksnumbered,
            colorlinks=true,
            citecolor=Blue,
            linkcolor=Blue,
            urlcolor=Blue,
            breaklinks,
            pdfborder=0 0 0,
            pdftitle=\PaperTitle,
            pdfauthor={Rémi Stroh, Julien Bect, Séverine Demeyer,
                       Nicolas Fischer, Damien Marquis, Emmanuel Vazquez},
            pdfkeywords={\PaperKWA, \PaperKWB}}

\usepackage{fancyhdr}
\newcommand{\spacingset}[1]{%
  \renewcommand{\baselinestretch}%
  {#1}\small\normalsize}

\newcommand \Rset  {\mathbb{R}}
\newcommand \Nset  {\mathbb{N}}
\newcommand \Tset  {\mathbb{T}}
\newcommand \Uset  {\mathbb{U}}
\newcommand \Xset  {\mathbb{X}}
\newcommand \ddiff {\mathrm{d}}                    
\newcommand{\D}[1]{\ddiff{#1}}                     
\newcommand \du   {\ddiff u}

\newcommand \dx   {\ddiff x}
\newcommand*{\Fenv}{P_{u_{\rm e}}}    
\newcommand*{\zcrit}{z^{\mathrm{crit}}}            
\newcommand*{\deltaref}{\delta^{\mathrm{ref}}}     
\newcommand \MedErr {\mathrm{MedErr}}
\newcommand \one    {\mathds{1}}
\newcommand \xx     {\tilde x}

\DeclareMathOperator  \prob   {\mathsf{P}}        
\DeclareMathOperator  \esp    {\mathsf{E}}        
\DeclareMathOperator  \Ncal   {\mathcal{N}}       
\DeclareMathOperator  \N      {\mathcal{N}}       
\DeclareMathOperator  \GP     {\mathsf{GP}}       
\DeclareMathOperator  \matern {\mathcal{M}}       
\DeclareMathOperator  \var    {var}               
\DeclareMathOperator  \cov    {cov}               
\DeclareMathOperator* \argmin {argmin}            
\DeclareMathOperator* \argmax {argmax}            

\newtheorem{lemma}{Lemma}
\newtheorem{proposition}[lemma]{Proposition}
\newtheorem{corollary}[lemma]{Corollary}
\theoremstyle{remark}
\newtheorem{remark}{Remark}

\begin{document}

\renewcommand{\thefootnote}{\fnsymbol{footnote}}  

\title{\bf\PaperTitle}

\author{%
  R\'emi \textsc{Stroh}$^{*\dagger}$ %
  \qquad Julien \textsc{Bect}$^{\star}$ %
  \qquad S\'everine \textsc{Demeyer}$^{\dagger}$\\[0.5em]
  Nicolas \textsc{Fischer}$^{\dagger}$ %
  \qquad Damien \textsc{Marquis}$^{\ddagger}$ %
  \qquad Emmanuel \textsc{Vazquez}$^{\star}$\\
  \\
  \small $^{\star}$ Laboratoire des signaux et syst\`emes, CentraleSup\'elec,\\[-2pt]
  \small Univ. Paris-Sud, CNRS, Universit\'e Paris-Saclay, France\\[-2pt]
  \small Email: firstname.lastname@l2s.centralesupelec.fr\\[-2pt]
  \\
  \small $^{\dagger}$ Département sciences des données et incertitudes\\[-2pt]
  \small $^{\ddagger}$ Département comportement au feu et sécurité incendie\\[-2pt]
  \small Laboratoire national de m\'etrologie et d'essais, Trappes, France\\[-2pt]
  \small Email: firstname.lastname@lne.fr
}

\date{}  \maketitle  \thispagestyle{fancy}

\begin{abstract}
  This article deals with the sequential design of experiments for
  (deterministic or stochastic) multi-fidelity numerical simulators,
  that is, simulators that offer control over the accuracy of
  simulation of the physical phenomenon or system under study.
  Very often, accurate simulations correspond to high computational
  efforts whereas coarse simulations can be obtained at a smaller
  cost.
  In this setting, simulation results obtained at several levels of
  fidelity can be combined in order to estimate quantities of interest
  (the optimal value of the output, the probability that the output
  exceeds a given threshold\ldots) in an efficient manner.
  To do so, we propose a new Bayesian sequential strategy called
  Maximal Rate of Stepwise Uncertainty Reduction (MR-SUR), that
  selects additional simulations to be performed by maximizing the
  ratio between the expected reduction of uncertainty and the cost of
  simulation.
  This generic strategy unifies several existing methods, and provides
  a principled approach to develop new ones.
  We assess its performance on several examples, including a
  computationally intensive problem of fire safety analysis where the
  quantity of interest is the probability of exceeding a tenability
  threshold during a building fire.
\end{abstract}

\vfill

\noindent%
\PaperKWB
\vfill

\newpage \spacingset{1.4}
\makeatletter
  \renewcommand*\l@section{\@dottedtocline{1}{1em}{3em}}
  \renewcommand*\l@subsection{\@dottedtocline{2}{4em}{4em}}
\makeatother
\tableofcontents

\newpage \spacingset{1.1}
\section{Introduction}
\label{sec:introduction}
\begingroup%
\renewcommand{\thefootnote}{}%
\footnote{%
  A preliminary version of this work was presented in
  \citet{stroh2017sequential}.
  The results contained in this article also appear in the PhD thesis
  of the first author \cite{stroh:thesis}.
}
\endgroup
In the domain of computer experiments, \emph{multi-fidelity} refers to
the idea of combining results from numerical simulations at different
levels of accuracy, %
high-fidelity simulations corresponding to more accurate but, in
general, more expensive computations.
As a representative example of a multi-fidelity simulator, consider
the case of a partial differential equation (PDE) solver based on a
finite element method:
the accuracy of the numerical solution depends among other things on
the fineness of the discretization. High fidelity results are
obtained when the mesh size is small.
Conceptually, a simulator is viewed in this article as a \emph{black box}
with inputs and outputs. %
The parameter that controls the level of accuracy/fidelity---the mesh
size in the case of a PDE solver---is one of the inputs of this black
box, %
alongside others such as design or control variables and environmental
variables \citep[see, e.g.,][]{santner2003design}.
Examples of multi-fidelity simulators can be found in virtually all
areas of engineering and science, including aeronautics
\citep{forrester2007multi}, fire safety \citep{demeyer2017surrogate},
electromechanics \citep{hage2014radial}, electromagnetism
\citep{koziel2013robust}, h\ae{}modynamics \citep{perdikaris2016model}
and many more.

When the objective is to estimate a particular quantity of interest
(QoI), %
such as the optimal value of the design variables (optimization problem)
or the probability that the outputs belong to a prescribed ``safe
region'' (reliability problem), %
multi-fidelity makes it possible to obtain a good approximation of the
QoI with a computational effort lower than what would have been
necessary if only high fidelity simulations had been carried out.
This cost reduction is achieved through the joint use of
multi-fidelity models, which allow simulation results obtained at
different levels of fidelity to be combined, and multi-fidelity
designs of experiments (DoE); see \cite{fernandez2016review} for a
review that covers both aspects.
This article addresses the problem of constructing, sequentially, a
multi-fidelity DoE targeting a given QoI.

We adopt a Bayesian point of view following the
line of research initiated by \citet{sacks1989design}---see also
\citet{currin1991bayesian}, \citet{santner2003design}\ldots---where
prior belief about the simulator is modeled using a
Gaussian process.
The Bayesian approach provides a rich framework for the construction
of sequential DoE, %
which has been abundantly relied upon in previous works dealing with
the case of single-fidelity simulators, where the cost of a simulation
is assumed to be independent of the value of the input variables %
\citep[see, e.g.,][]{kushner64, mockus78, jones1998efficient,
  ranjan2008, villemonteix2009informational, picheny2010adaptive,
  bect2012sequential, chevalier2014fast}.
In this framework, sequential designs are usually constructed by means
of a sampling criterion---also called infill criterion, acquisition
function or merit function---, the value of which indicates
whether a particular point in the input space is promising or not.
The expected improvement (EI) criterion \citep{jones1998efficient} is
a popular example of such a sampling criterion.
The extension of the Bayesian approach to sequential DoE in a
multi-fidelity setting is based on two
ingredients: 1) the construction of prior models for simulators with
adjustable fidelity; 2) the construction of sampling
criteria that take the variable cost of simulations into account.

For the case of \emph{deterministic} multi-fidelity simulators, Gaussian
process-based models have already been proposed in the literature
\citep{kennedy2000predicting, le2013multi, le2014recursive,
  picheny2013nonstationary, tuo2014surrogate}, %
Extensions to \emph{stochastic} simulators have been proposed as well
\citep{stroh2016gaussian, stroh2017assessing}.

Sampling criteria for single-fidelity sequential designs do not
reflect a crucial feature of multi-fidelity simulators: the cost of a
run depends on the value of the inputs (in particular on the one that
controls the fidelity of simulation). Various methods that take into
account the variable cost of the simulations have been proposed for
particular cases, for single-objective unconstrained optimization
\citep{huang2006sequential, swersky2013multi, he2017optimization} and
global approximation \citep{xiong2013sequential, gratiet2015kriging},
notably.

In this article, we provide a general principled methodology to
construct sequential DoE for multi-fidelity simulators %
and, more generally, for simulators where the cost of a simulation
depends on the value of the inputs.
The methodology is applicable to any QoI, and builds on the \emph{Stepwise
Uncertainty Reduction} (SUR) principle \citep[see, e.g.,][and
references therein]{villemonteix2009informational, bect2012sequential,
  chevalier2014fast, bect2016supermartingale}, which unifies many of
the aforementioned sequential DoE for the fixed-cost case.
More precisely, for the variable-cost case, we propose the \emph{Maximal
Rate of Stepwise Uncertainty Reduction} (MR-SUR) principle, %
which consists in constructing a sequential design by maximizing, at
each step, the ratio between the expected reduction of uncertainty %
(to be defined more precisely later on) and the cost of the
simulation.

The article is organized as follows.
Section~\ref{sec:models} reviews Gaussian process modeling for
deterministic simulators, and discuss some possible extensions to
(normally distributed) stochastic simulators.
Section~\ref{sec:doe} first reviews both existing methods of
sequential design for multi-fidelity simulators and the SUR principle
for fixed-cost simulators, and then presents the MR-SUR principle and
its relations with some existing sequential DoE.
Finally, Section~\ref{sec:illustration} illustrates the method and
assesses its performance through several academic examples, %
including a computationally intensive problem of fire safety analysis
where the quantity of interest is the probability of exceeding a
tenability threshold during a building fire.

\section{Gaussian-process models for multi-fidelity}
\label{sec:models}

We consider a computer simulator with input variables
$u \in \Uset \subset \Rset^d$ and one or several scalar outputs, which
are generally obtained after some post-processing steps (e.g., an
aerodynamic drag in a CFD model).
Moreover, we consider that the accuracy, or \emph{fidelity}, of the
computer simulation can be tuned using a parameter $\delta$
that ranges in a discrete or continuous set~$\Tset$.
For instance, $\delta$ is a mesh size in a finite element method.
Such a parameter will be called \emph{fidelity parameter} and can be
viewed as an additional input of the simulator.
We denote by $x = (u, \delta) \in \Xset$ the aggregated vector of
inputs, with $\Xset = \Uset \times \Tset$,
and we assume from now on that the output is scalar.

\subsection{The auto-regressive model for deterministic simulators}
\label{subsec:model:KO}

The so-called \emph{auto-regressive model} of
\citet{kennedy2000predicting} assumes a deterministic simulator with a
finite number~$S$ of levels of increasing fidelity.
Let $\delta_1,\,\ldots, \delta_S$ denote the corresponding values of
the fidelity parameter and set
$\Tset = \{ \delta_1, \ldots, \delta_S \}$.
The simulator is then modeled by a Gaussian process~$\xi$
on~$\Xset = \Uset \times \Tset$,
defined through an auto-regressive relationship between successive
levels:
\begin{equation}
  \left\{
    \begin{aligned}
      \xi(u, \delta_1) & = \eta_1(u),\\
      \xi(u, \delta_s) & = \rho_{s-1}\, \xi(u, \delta_{s-1}) + \eta_s(u),
      \quad 1 < s \le S,
    \end{aligned}\right.
\end{equation}
where $\eta_1$, \ldots, $\eta_S$ are $S$~mutually independent Gaussian
processes, and $(\rho_s)_{1\leq s < S} \in \Rset^{S - 1}$.

The model has been used in numerous applications, where the actual
number~$S$ of levels is most often two
\citep[see, e.g.,][]{%
  forrester2007multi, kuya2011multifidelity, brooks2011multi,
  wankhede2012multi, goh2013prediction, le2014recursive,
  gratiet2015kriging, elsayed2015optimization, thenon2016multi,
  demeyer2017surrogate},
sometimes three \citep[][Section~3.2]{perdikaris2016model}.
In practice, the Gaussian processes~$\eta_s$ are chosen among a family
of Gaussian processes indexed by (hyper-)parameters such as
correlation lengths, regularity parameters, etc., which are estimated
from data (simulation results), by maximum likelihood for instance
(see, e.g., \cite{stein1999interpolation}).
Since the processes $\eta_s$ are assumed independent, there must be
enough simulation results at each level of fidelity, even at the
possibly very expensive highest fidelity levels, to obtain good
estimates of the hyper-parameters---%
which explains perhaps why this model is typically used with a small
number of levels.

\subsection{The additive model for deterministic simulators}
\label{subsec:model:TWY}

Another approach to building Gaussian process models for deterministic
multi-fidelity simulators, which readily applies to the case where
$\Tset$ contains a continuum of levels of fidelity or a large number
of ordered discrete levels of fidelity, has been proposed by
\citet{picheny2013nonstationary} and \citet{tuo2014surrogate}.

Assuming for simplicity that $\Tset = \left[0, \infty\right)$,
with $\delta = 0$ corresponding as in \citet{tuo2014surrogate} to the
highest---often unreachable---level of fidelity,
a Gaussian process~$\xi$ over the product
space~$\Xset = \Uset \times \Tset$ is defined in this approach as the
sum of two independent parts:
\begin{equation}
  \label{eq:TWY:sum}
  \xi(u, \delta) = \xi_0(u) + \varepsilon(u, \delta),
\end{equation}
where $\xi_0$ and $\varepsilon$ are mutually independent Gaussian
processes, and $\varepsilon$ has zero mean and goes to zero in the
mean-square sense when $\delta \to 0$.
In other words, $\var \left( \varepsilon (u, \delta) \right) \to 0$
when $\delta \to 0$, for all~$u$: as a consequence, $\xi$ is a
\emph{non-stationary} Gaussian process on $\Xset \subset \Rset^{d+1}$.
Under this decomposition, $\xi_0$ represents the ``ideal'' version of
the simulator, while $\varepsilon$ represents numerical error.

In both articles, $\xi_0$ is then assumed to be stationary, whereas
the covariance function of $\varepsilon$ is multiplicatively
separable:
for all $u, u' \in \Uset$ and $\delta, \delta' \in \Tset$,
\begin{equation}
  \label{eq:TWY:cov-epsi}
  \cov\left(\varepsilon(u, \delta), \varepsilon(u', \delta')\right)
  = r(\delta, \delta')\, k(u, u'),
\end{equation}
where $k$ is a stationary covariance function on~$\Uset$, and $r$ is a
(non-stationary) covariance function on~$\Tset$ such that
$r(\delta, \delta) \to 0$ when $\delta \to 0$.
As an example of a suitable choice for~$r$, consider the Brownian-type
model proposed by \citet{tuo2014surrogate}:
\begin{equation}
  \label{eq:TWY:Brownian-cov}
  r(\delta, \delta') = \min\{\delta, \delta'\}^L,
\end{equation}
with $L$ a real positive parameter.
Other choices are of course possible.

\newcommand \EquModelTWY {\eqref{eq:TWY:sum}--\eqref{eq:TWY:Brownian-cov}}

\subsection{Extension to stochastic simulators}
\label{subsec:model:stoch}

We now turn to the case of stochastic simulators, that is, simulators
whose output is stochastic, %
as happens for instance when the computer program relies on a Monte
Carlo method \citep[see, e.g.,][]{cochet2014}.
Extending the multi-fidelity Bayesian methodology of
Sections~\ref{subsec:model:KO} and~\ref{subsec:model:TWY} to
stochastic simulators is not straightforward in general, %
since the output at a given input point~$x_i = (u_i, \delta_i) \in \Xset$ is
now a random variable~$Z_i$, the distribution of which is in general
unknown and different at each point in~$\Xset$.
(Several runs at the same input point yield independent and
identically distributed responses.)

We focus in this section on the simpler case where the output~$Z_i$
can be assumed to be normally distributed:
\begin{equation}
  Z_i \,\vert\, \xi, \lambda \;\sim\; \N(\xi(x_i), \lambda(x_i)),
  \label{eq:normal_outputs}
\end{equation}
with mean~$\xi(x_i)$ and variance~$\lambda(x_i)$ possibly depending on the
input point.
In this setting, we propose to extend the multi-fidelity models of
previous sections using independent prior distributions for~$\xi$
and~$\lambda$, %
with either the autoregressive model of Section~\ref{subsec:model:KO}
or the additive model of Section~\ref{subsec:model:TWY} as a prior
for~$\xi$.
Then, since $\lambda$ must have positive values and we want to retain
the simplicity of the Gaussian process framework, %
we suggest modeling the logarithm of the variance,
i.e. $\log(\lambda)$, by a Gaussian process~$\tilde \lambda$,
following \citet{goldberg1998regression}, \citet{kersting2007},
\citet{boukouvalas2009learning} and others.

Under this type of model, the inference task%
---estimating the hyper-parameters of the Gaussian process models
for~$\xi$ and~$\tilde{\lambda}$, and computing posterior
distributions---%
becomes more difficult since neither~$\xi$ nor~$\tilde{\lambda}$ are
directly observable.
\cite{goldberg1998regression} take a fully Bayesian approach and
suggest using a time-consuming Monte-Carlo method.
Other authors have proposed optimization-based approaches, that
simultaneously produce estimates of both the Gaussian processes
hyper-parameters and the unobserved log-variances:
in particular, \citet{kersting2007}
and~\citet{boukouvalas2009learning} propose a method called \emph{most
  likely heteroscedastic GP}, stemming from the
Expectation-Maximization (EM) algorithm %
\citep[see also][for a similar algorithm]{marrel2012global},
while \citet{binois2018} use a more a more sophisticated joint
maximization procedure with relaxation to obtain the joint MAP
(maximum a posterior) estimator.

For the numerical experiments of this article
(Sections~\ref{subsec:illustration_dampedOscillator}
and~\ref{subsec:illustration_fds}) we will take a simpler route,
assuming that the variance~$\lambda$ depends only on the fidelity
level~$\delta$---which is approximately true in the two examples we
shall consider.
In this setting, as long as the number of fidelity levels of interest
is not too large, the value of the variance at these levels can be
simply estimated jointly with the other hyper-parameters of the model;
a general-purpose log-normal prior for the vector of variances is
proposed by \citet{stroh2016gaussian, stroh2017integrating}.

\section{Sequential design of experiment for multi-fidelity}
\label{sec:doe}

\subsection{Existing methods}
\label{subsec:doe_multi}

In the literature of multi-fidelity, a variety of sequential design
algorithms have been proposed.
(See Supplementary Material for a review of \emph{non-sequential}
multi-fidelity designs, which can be used as initial designs for
sequential ones.)

For instance, \citet{forrester2007multi} suggest using the
auto-regressive model of \cite{kennedy2000predicting} and a standard
single-level sequential design at the highest level of fidelity to
select input variables $u\in\Uset$ for the next experiment.
Then, simulations at all levels of fidelity are run for the selected
$u$.
Building on \citet{forrester2007multi}, \citet{kuya2011multifidelity}
suggest a two-stage method: run a large number of simulations at the
low-fidelity level, and then use a sequential design strategy to
select simulations at the high-fidelity level.
In a different spirit, \citet{xiong2013sequential} use Nested Latin
Hypercube Sampling (NLHS) and suggest to double the number of
simulations when going from a level~$\delta^{(s)}$
to~$\delta^{(s + 1)}$, until some cross-validation-based criterion is
satisfied.

More interestingly in the context of this article, some methods have
been proposed that explicitly take into account the simulation cost.
This is typically achieved by crafting a sampling
criterion that takes the form of a ratio between a term which measures
the interest of a simulation at $(x, \delta)$, and the cost of the
simulation \citep{huang2006sequential,
  gratiet2015kriging, he2017optimization}.
For instance, \cite{he2017optimization} propose a global optimization
method using the Expected Quantile Improvement (EQI) of
\cite{picheny2013quantile} and the multi-fidelity model of
\cite{tuo2014surrogate}, and build a new sampling criterion
corresponding to the ratio between the EQI sampling criterion and the
cost of a simulation.

Outside the multi-fidelity literature, a similar idea has been
proposed by \citet{johnson1960information} to design sequential
testing procedures and by \citet{swersky2013multi} for multi-task
optimization.
In both cases, the numerator of the criterion is the expected
reduction of the entropy of the QoI.

In this article, we propose a general methodology to build such
sequential designs, which is not tied to a particular kind of model or
QoI.
The key idea is to measure the potential of a particular design point
using the SUR framework, recalled in Section~\ref{subsec:doe_sur}.
The methodology itself, that we call MR-SUR, is presented in
Sections~\ref{subsec:doe_mrsur}.

\subsection{Stepwise Uncertainty Reduction}
\label{subsec:doe_sur}

We recall here the principle of SUR strategies, introduced in the
design of computer experiments by Vazquez and co-authors
\citep[][\ldots]{vazquez:2007:ds, villemonteix2009informational,
  vazquez2009sequential, bect2012sequential, chevalier2014fast}.
Given a Bayesian model of a simulator and an unknown QoI $Q$,
that is, a particular feature of the simulator that we want to
estimate, a SUR strategy is a Bayesian method for the construction
of a sequence of evaluation locations $X_1, X_2, \ldots\in\Xset$ at
which observations of the simulator will be taken in order to reduce
the uncertainty on $Q$.
(In this section, $\Xset$ denotes a generic input space, not
necessarily of the form~$\Xset = \Uset \times \Tset$.)

The starting point of the construction of a SUR strategy is the
definition of a statistic~$H_{n}$ measuring the residual uncertainty
about $Q$ given past observations $Z_1, \ldots, Z_n$.
Many choices for~$H_n$ are possible for any particular problem, but a
natural requirement \citep{bect2016supermartingale} is that~$H_n$
should be decreasing on average when $n$~increases.
For instance, if~$Q$ is a scalar QoI, $H_n$ could be the posterior
entropy or the posterior variance of~$Q$.
If $Q$ is a function defined on~$\Xset$, as will be the case in
Section~\ref{sec:illustration}, a possible choice is
\begin{equation}
  H_n \;=\; %
  \esp_n \left( \lVert Q - \widehat Q_n \rVert_\mu^2 \right) \;=\; %
  \int_\Xset \var_n \left( Q(x) \right)\, \mu(\dx),
  \label{equ:Hn-L2}
\end{equation}
where $\mu$ denotes a measure~$\Xset$, %
$\lVert h \rVert_\mu^2 = \int_\Xset h(x)^2\, \mu(\dx)$,
$\esp_n$ (resp.~$\var_n$) is the posterior expectation
(resp.~variance) given~$Z_1$, \ldots, $Z_n$, %
and $\widehat Q_n (x) = \esp_n \left( Q(x) \right)$.

Then, given past observations, $X_{n+1}$ is chosen by minimizing the
expectation of the future residual uncertainty:
\begin{equation}
  X_{n+1}= \argmin_{x\in \Xset}
  J_n\left(x\right),\quad\text{with }
  J_n\left(x\right) = \esp_n\left(H_{n + 1} \middle \vert
    X_{n+1} = x\right),
  \label{eq:sur_crit}
\end{equation}
where the expectation is with respect to the outcome~$Z_{n+1}$ of a
new simulation at~$x \in \Xset$.

\medbreak

\noindent\textbf{Example.}
Assume a stochastic multi-fidelity simulator defined over
$\Xset = \Uset \times \Tset$ as in Section~\ref{subsec:model:stoch},
and consider the functional QoI defined on~$\Xset$ by
\begin{equation}
  \label{eq:QoI-prob-fun}
  Q(x) \;=\; %
  \prob \left( Z_x > \zcrit \bigm| \xi, \lambda \right) \;=\; %
  \Phi \left( \frac{ \xi(x) - \zcrit}{\sqrt{\lambda(x)}} \right),
\end{equation}
where $Z_x$ denotes the outcome of a new simulation at~$x$,
$\zcrit \in \Rset$ is a given threshold, and $\Phi$ the cdf of the
standard normal distribution.
Pick some reference level~$\deltaref \in \Tset$ and consider the
residual uncertainty
\begin{equation}
  H_n \;=\; %
  \int_{\Uset} \var_n \left( Q(u, \deltaref) \right)\, \du,
  \label{equ:Hn-L2-U}
\end{equation}
which is a special case of~\eqref{equ:Hn-L2} with $\mu$ equal to
Lebesgue's measure on~$\Uset$ at fixed $\delta = \deltaref$.
Then, using computations similar to those of
\citet{chevalier2014fast}, it can be proved that
\begin{equation}
  J_n(x) = \int_{\Uset} \left[
    \Phi_2\Bigl(a_n(x'), a_n(x');
    \frac{k_n(x', x')}{v_n(x')}\Bigr)
    - \Phi_2\Bigl(a_n(x'), a_n(x');
    \frac{k_n(x, x')^2}{v_n(x) v_n(x')}\Bigr)
  \right]\, \ddiff u',
  \label{eq:notre-SUR-crit}
\end{equation}
where $x' = (u', \deltaref)$, $m_n$ (resp.~$k_n$) denotes the
posterior mean (resp.~covariance) of~$\xi$,
$v_n(x) = \lambda(x) + k_n(x, x)$,
$a_n(x) = \left(m_n(x) - \zcrit\right)/\sqrt{v_n(x)}$, and
$\Phi_2 \left( \cdot, \cdot \,; \rho \right)$ is the cdf of the standard
bivariate normal distribution with correlation~$\rho$.
(For tractability, the variance function $\lambda$ is assumed to be
known in the computation of the criterion.  In practice, the estimated
variance function is plugged in the expression, and the integral
over~$\Uset$ is approximated using a Monte Carlo method.)

\begin{remark}
  See Supplementary Material for a proof of~\eqref{eq:notre-SUR-crit},
  in a more general form which also allows for batches of parallel
  evaluations and integration of~$Q$ with respect to environmental
  variables (all or part of the components of~$u$, depending on the
  application).
\end{remark}

\begin{remark} \label{rem:special-case}
  In the special case~$\lambda \equiv 0$ (deterministic simulator),
  corresponding to $Q(x) = \mathds{1}_{\xi(x) > \zcrit}$, the
  criterion~\eqref{eq:notre-SUR-crit} has been proposed by
  \cite{bect2012sequential} and computed by \citet{chevalier2014fast}.
  The general case is new, to the best of our knowledge.
\end{remark}

The reader is referred, e.g., to \cite{villemonteix2009informational,
  picheny2010adaptive, chevalier13, chevalier2014fast,
  bect2016supermartingale} for other examples of SUR criteria.

\subsection{Maximum Rate of Stepwise Uncertainty Reduction}
\label{subsec:doe_mrsur}

The proposed Maximum Rate of Stepwise Uncertainty Reduction (MR-SUR)
strategy builds on the SUR strategy presented in
Section~\ref{subsec:doe_sur}.
The goal is to achieve a balance between the (expected) reduction of
uncertainty brought by new observations on the one hand, and the cost
of these observations, usually measured by their computation time, on
the other hand.
Denoting by $C:\Xset \to \Rset_{+}$ the cost of an observation of the
simulator, which depends on the fidelity level $\delta \in \Tset$
and/or input variables $u\in\Uset$, the MR-SUR strategy is given by
\begin{equation}
  X_{n+1} \;=\; %
  \argmax_{x \in \Xset}\, \frac{H_n - J_n(x)}{C(x)}
  \;=\; %
  \argmax_{x \in \Xset}\, \frac{G_n(x)}{C(x)}\,,
  \label{eq:mrsur}
\end{equation}
where $G_n(x) = H_n - J_n(x)$ is the \emph{expected uncertainty
  reduction} associated to a future observation at~$x \in \Xset$.
This strategy boils down to a SUR strategy when $C$ is constant.

A few special cases of MR-SUR strategies, adapted to particular models
and estimation goals, have been proposed earlier in the literature.
To the best of our knowledge, the oldest example is the sequential
testing method of \cite{johnson1960information}, where~$H_n$ is the
posterior entropy of the location of faulty component in an electronic
equipment---with a discrete distribution over all possible fault
locations as the underlying model.
More recently, \cite{snoek2012practical} and \cite{swersky2013multi}
have proposed Bayesian optimization procedures of the MR-SUR type, for
unconstrained global optimization problems with variable-cost
noiseless evaluations, corresponding respectively, when to cost is
constant, to the expected improvement \citep{mockus78,
  jones1998efficient} and IAGO \citep{villemonteix2009informational}
algorithms.
Finally, the first sequential design procedure of
\cite{gratiet2015kriging} can also be seen as an approximate MR-SUR
strategy for the approximation problem, where $H_n$ is the posterior
integrated prediction variance.

To illustrate the MR-SUR principle, let us consider a simple simulated
example,
with $\xi$ a Gaussian process on
$\Xset = \Uset \times \Tset = [-0.5,0.5]\times[0,1]$ such that
$\xi \mid m \sim \GP(m, k)$,
$m \sim \mathcal{U}(\Rset)$,
and $k$ as in Section~\ref{subsec:model:TWY}:
\begin{equation*}
  k:((u,\delta), (u^{\prime}, \delta^{\prime}))
  \;\mapsto\; %
  \sigma_0^2 \mathcal{M}_{\nu_0} \left(\frac{|u - u^{\prime}|}{\rho_0}\right)
  \,+\, %
  \sigma_0^2G\min\left\{\delta, \delta^{\prime}\right\}^L
  \mathcal{M}_{\nu_{\varepsilon}}
  \left(\frac{|u - u^{\prime}|}{\rho_{\varepsilon}}\right)\,,
\end{equation*}
where $\mathcal{M}_{\nu}$ stands for the Matérn correlation function
with regularity parameter $\nu$.
The values $m = 0$, $\sigma_0 = 1$, $G = 4$, $L = 2$,
$\nu_0 = \nu_{\varepsilon} = 5/2$, $\rho_0 = 0.3$,
$\rho_{\varepsilon} = 0.1$ are used in the simulations, and all the
parameters except $m$ are assumed to be known in this experiment.
The cost function is $C: (u, \delta) \mapsto 1/\delta$ and the QoI is
\begin{equation*}
  Q =\int_{\Uset} \mathds{1}_{\xi(u,0) > 0}\, \du\,.  
\end{equation*}
Note that the level of highest fidelity $\delta=0$ is not observable
in practice.
A NLHS of size $n = 12 + 6 + 6 + 3$ on the levels
$\delta = 1, 1/2, 1/5, 1/10$ is taken as the observed DoE,
and the outputs~$Z_1, \ldots, Z_n$ are simulated according
to~\eqref{eq:normal_outputs} with constant variance~$\lambda = 0.4^2$.
We compute the functions $J_n$, $G_n$ and $C$ over a regular grid on
$\Uset \times \Tset$, to obtain Figures~\ref{fig:benef_vs_cost}
and~\ref{fig:mean_benef}.

\begin{figure}
  \begin{center}
    \psfrag{Cost}[tc][tc]{Cost of an observation $C(x)$}
    \psfrag{Expected uncertainty}[bc][bc]{$J_n(x)$}
    \psfrag{Pareto line}[bc][bc]{}
    \subfloat[Uncertainty and cost]{
      \includegraphics[width = 0.5\textwidth]{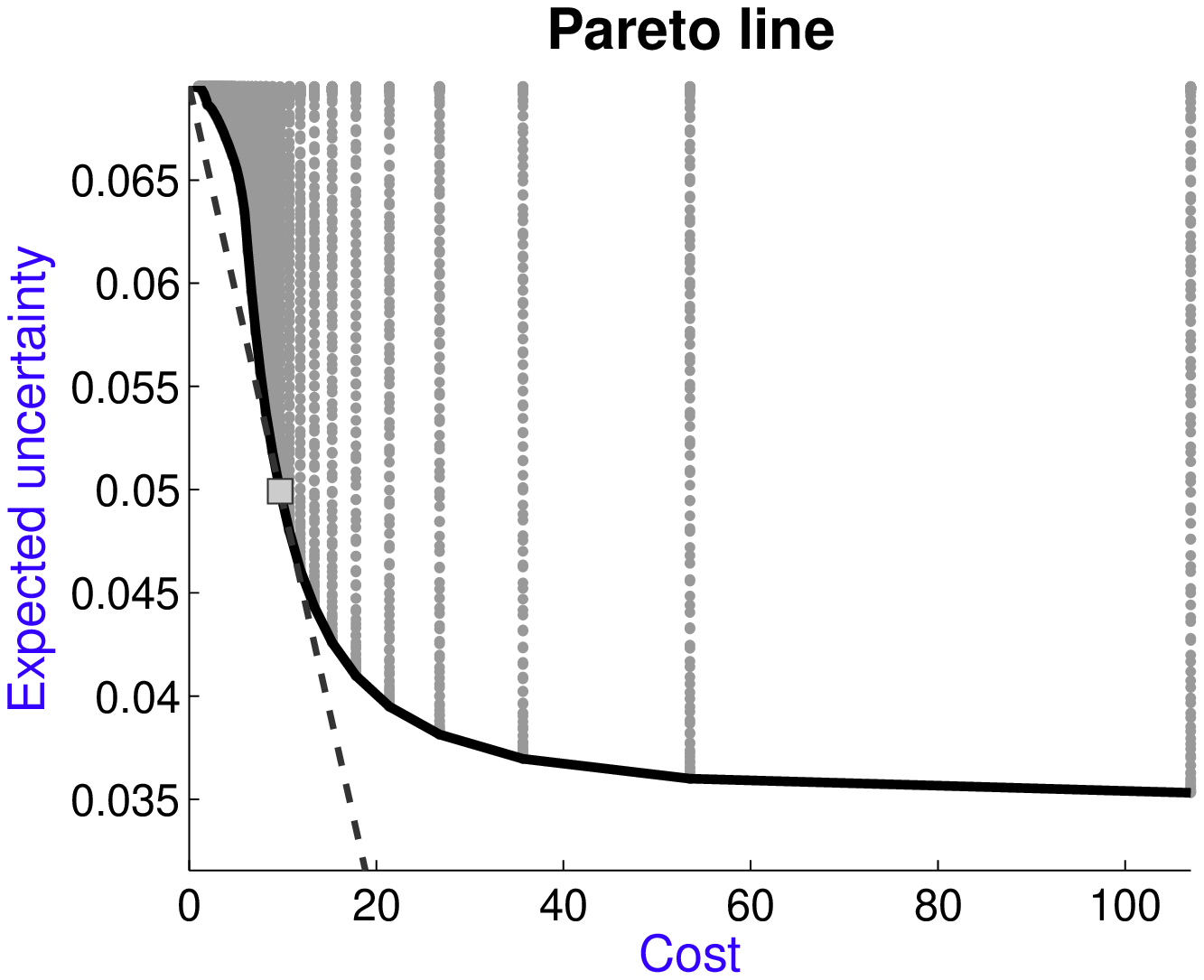}
      \label{fig:benef_vs_cost}}
    \psfrag{Cost}[tc][tc]{Cost of an observation $C(x)$}
    \psfrag{Benefit./Cost}[bc][bc]{$[H_n - J_n(x)]/C(x)$}
    \psfrag{Mean benefit}[bc][bc]{}
    \subfloat[Gain-cost ratio]{
      \includegraphics[width = 0.5\textwidth]{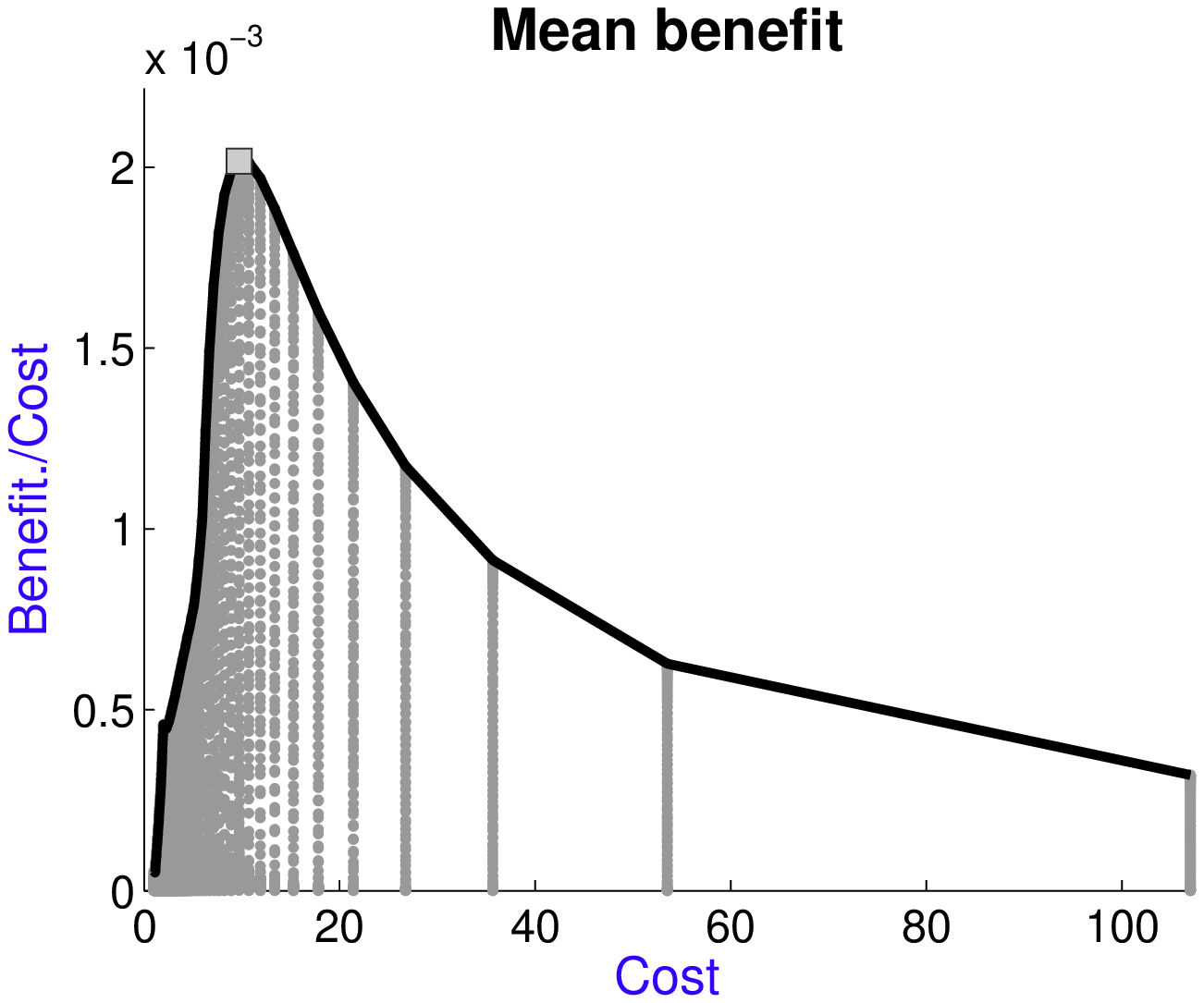}
      \label{fig:mean_benef}}
  \end{center}
  \caption{%
    An example of objective space.%
    (a) Representation of possible designs in the $(C, J)$ plane. %
    (b) Representation in the $(C, G/C)$ plane. %
    Each point corresponds to one point $x$, the solid line is the
    Pareto-optimal points, and the square is the design returned by
    the Maximal Rate of Uncertainty Reduction criterion.%
  }
\end{figure}

Observe on Figure~\ref{fig:benef_vs_cost} that, for each cost value
(corresponding to a fixed fidelity level), there is a range of points
that yield more or less expected uncertainty reduction.
Good observation points lie on the Pareto front (in solid black line),
that is, the set of points for which there is no larger expected
uncertainty reduction at lower cost.
The MR-SUR strategy selects an observation location that correspond to
the maximum of the ``slope'' of the Pareto front.

Figure~\ref{fig:example_seqdes} shows the sequence of Pareto fronts as
more observation points are added in the design using~\eqref{eq:notre-SUR-crit}.
The horizontal axis is the total cost, so that the left-ends of the
Pareto fronts are shifted.
Observe for instance that the points numbered~$3$ to~$9$, selected
using MR-SUR, achieve a larger uncertainty reduction at lower cost that
what would have been achieved if we had selected only one expensive
observation.

\begin{figure}
  \begin{center}
    \psfrag{Cost}[tc][tc]{Total cost}
    \psfrag{Expected uncertainty}[bc][bc]{$J_n(x)$}
    \psfrag{Pareto line}[bc][bc]{}
    
    \includegraphics[width = 0.7\textwidth]{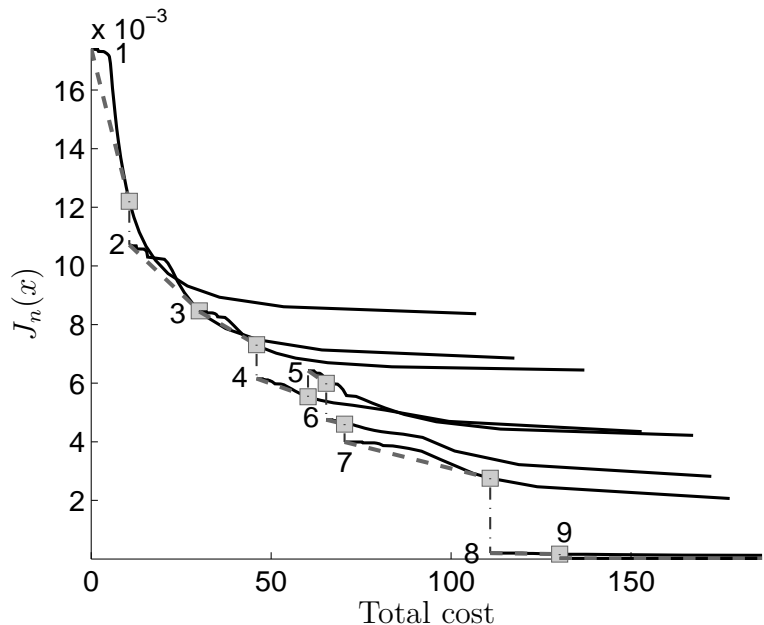}
  \end{center}
  \caption{%
    The sequential Pareto fronts in the space $(C, J)$ as function of
    the total cost of the design on an example of sequential MR-SUR
    algorithm.%
  }
  \label{fig:example_seqdes}
\end{figure}

\section{Numerical results}
\label{sec:illustration}

\subsection{Setup of the experiments}

In each example, we consider a multi-fidelity simulator for which
simulation cost $C$ depends on~$\delta$ alone, and is assumed to be
known.
Some common features of all three numerical experiments are presented
in this section.

\emph{Initial DoE.}
A nested Latin hypercube sample (NLHS) is used as an initial design.
More specifically, we use the algorithm developed by
\cite{qian2009nested}, with an additional maximin optimization at each
level to obtain better space-filling properties
\citep[see][Section~2.2.3 for details]{stroh:thesis}.

\emph{GP modeling.}
In each example, a multi-fidelity GP model of the type described in
Section~\ref{sec:models} is used.
The posterior distribution of the parameters is initially sampled
using an adaptive Metropolis-Hastings algorithm
\citep{haario2001adaptive} and then updated at each iteration by
sequential Monte Carlo \citep[see, e.g.,][]{chopin2002sequential}.
More details about the particular GP model that is used, and the prior
distribution on the parameters, are provided inside each example
section.

\emph{Optimization of the sampling criterion.}
At each iteration of a SUR or MR-SUR strategy, a new simulation point
is selected according to~\eqref{eq:sur_crit} or~\eqref{eq:mrsur}.
This step involves a an optimization of the SUR or MR-SUR criterion
criterion, which is carried out in the experiments of this article
using a simple two-step approach:
the criterion is first optimized by exhaustive search on a regular
grid over $\Uset\times\Tset$, and then a local optimization is
performed starting from the best point in the grid.
Other approaches have been proposed in the literature, that would be
more efficient in higher-dimensional problems \citep[see,
e.g.,][]{Feliot:BMOO}.

\emph{Other computational details.}
All integrals are approximated by Monte-Carlo methods.
SUR and MR-SUR criteria are evaluated using the Maximum A Posteriori
(MAP) estimator of the parameters---obtained by local optimization from
the best point in the MCMC/SMC sample---in a plug-in manner.
(A fully Bayesian approach could be considered in principle, but would
lead to much higher computational complexity.)

\subsection{A one-dimensional example}
\label{sec:one-dim-example}

Consider as a first (toy) example the two-level deterministic
simulator defined for $u \in \left[ 0; 1 \right]$ and
$\delta \in \left\{ 1, 2 \right\}$ by the analytical formulas
\citep{forrester2007multi}
\begin{equation}
  \left\{
  \begin{aligned}
    f_1(u) &= f(u, 1) = 0.5\, (6u - 2)^2\sin(12u - 4) + 10\, (u - 0.5),\\
    f_2(u) &= f(u, 2) = (6u - 2)^2\sin(12u - 4) + 10,
  \end{aligned}
  \right.
\end{equation}
and assume that computing~$f_2$---hereafter referred to as the ``high
fidelity'' function---is four times more costly than computing~$f_1$,
e.g., $C(2) = 1$ and $C(1) = \frac{1}{4}$.
Note that the two functions are related by
$f_2(u) = 2\, f_1(u) - 20(u - 1)$,
which makes them perfect candidates for the autoregressive model
presented in Section~\ref{subsec:model:KO}.
The goal in this example is to estimate the set
\begin{equation*}
  \Gamma = \left\{ f_2 > \zcrit \right\}
  = \left\{ u\in \Uset,\, f_2(u) \geq \zcrit \right\}
\end{equation*}
with $\zcrit = 10$.
The performance of MR-SUR for this task will be compared with that
of SUR strategies operating at the low-fidelity level only (LF-SUR) or
at the high-fidelity level only (HF-SUR).

In this experiment, all three sequential strategies start with the
same multi-fidelity initial design, and use the same Gaussian process
prior and the same measure of uncertainty~$H_n$.
The initial design consists of six observations at the low-fidelity
level and three at the high-fidelity level,
for a total of~$n = 9$ observations, corresponding to an initial
budget of $6 \times \frac{1}{4} + 3\times 1 = 4.5$ cost units.
A supplementary budget of $9.0$~cost units is assumed to be
available for the sequential design.
The autoregressive model of Section~\ref{subsec:model:KO} is used,
with Matérn 5/2 covariance functions and weakly informative priors on
the parameters (see Supplementary Material for details).
The uncertainty on~$\Gamma$ is quantified using the uncertainty
measure~\eqref{equ:Hn-L2-U}.
In the special case of a deterministic simulator, we have
(cf. Remark~\ref{rem:special-case})
\begin{equation*}
  H_n = \int_0^1 p_n(u)\, \left(1 - p_n(u)\right) \du,
\end{equation*}
where $p_n(u) = \prob_n\left( \xi_2(u) \geq \zcrit \right)$ is the
posterior mean of~$\mathds{1}_{\xi_2(u) \geq \zcrit}$, and
$p_n(u)\, \left(1 - p_n(u)\right)$ its posterior variance.

\begin{figure}
  \quad
    \begin{subfloatrow}
      \subfloat[Comparison between designs of experiment]{%
        \label{fig:ex1_errors}
        \psfrag{Cost}[tc][tc]{$c_n$}%
        \psfrag{||p(x, tHF) - p*(x, tHF)||}[bc][bc]{$\MedErr_n$}%
        \psfrag{L2-error on the probability function}[bc][bc]{}%
        \includegraphics[width = 0.45\textwidth]{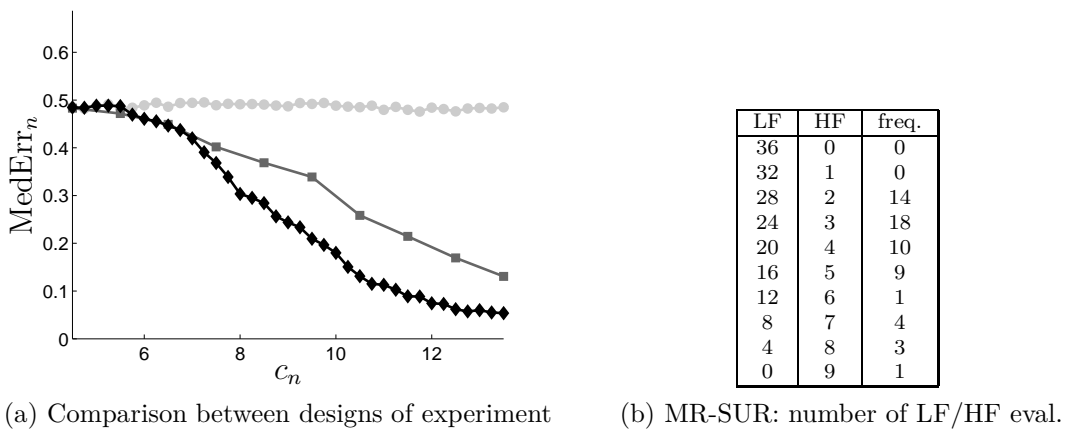}
      }
    \renewcommand{\arraystretch}{0.9}
    \subfloat[MR-SUR: number of LF/HF eval.]{%
      \scriptsize
      \hspace{6em}
      \begin{tabular}[b]{|c|c|c|}
        \hline
        LF & HF & freq.\\
        \hline
        36 & 0 & 0 \\
        32 & 1 & 0 \\
        28 & 2 & 14 \\
        24 & 3 & 18 \\
        20 & 4 & 10 \\
        16 & 5 & 9 \\
        12 & 6 & 1 \\
        8 & 7 & 4 \\
        4 & 8 & 3 \\
        0 & 9 & 1 \\
        \hline
      \end{tabular}
      \hspace{6em}
      \label{fig:ex1_histobs}%
    }
  \end{subfloatrow}
  \hfill
      \caption{%
      The one-dimensional experiment.  %
      (a) Median estimation error as a function of the cost.
      Light-gray disks: LF-SUR; %
      dark-gray squares: HF-SUR; %
      black diamonds: MR-SUR.
      (b) Number of LF/HF evaluations in the MR-SUR strategy.
      The last column indicates how many times, in the 60 repetitions,
      a given combination appears (recall that HF evaluations are
      four times as costly as LF ones).
    }
\end{figure}

The experiment is repeated $R = 60$ times---the simulator is
deterministic, but randomness in the result comes from both the
initial DoE and the use of a Monte Carlo procedure to sample from the
posterior of the parameters.
Figure~\ref{fig:ex1_errors} presents the evolution of the median
estimation error, defined as
$\MedErr_n = \mathrm{median}_{1 \le r \le R}\, \lVert p^{(r)}_n -
\one_{f_2 > \zcrit} \rVert$ with $\lVert \cdot \rVert$ the $L^2$-norm
on~$\Uset$, as function of the cost $c_n = \sum_{i \le n} C(\delta_i)$.
First, it appears clearly that high fidelity evaluations are needed:
the LF-SUR strategy achieves no significant error reduction with
respect to the initial design.
Second, we observe that the combination of low- and high-fidelity
evaluations chosen by the MR-SUR strategy is more efficient, on
average, than a purely high-fidelity sequential design.
The actual number of evaluations on each level is summarized, for the
60 repetitions, in Table~\ref{fig:ex1_histobs}: the MR-SUR strategy
tends to use between two and five high-fidelity evaluations.
(The recommendation of \citet{xiong2013sequential}---observing the
low-fidelity level twice as many times as the high-fidelity
one---would correspond here to six high-fidelity evaluations.)

\subsection{Random damped harmonic oscillator}
\label{subsec:illustration_dampedOscillator}

We now assess the performance of MR-SUR on an example proposed
by~\citet{au2001estimation}.
We consider a random damped harmonic oscillator, whose displacement
$X$ is the solution of the stochastic ordinary differential equation
\begin{equation}
  \ddot{X}(t) + 2\zeta\omega_0\dot{X}(t) + \omega_0^2X(t) = W(t),
  \quad t \in [0, t_{\mathrm{end}}], \quad \dot{X}(0) = 0,\quad X(0) = 0\,,
  \label{eq:ex2_stochas_eq}
\end{equation}
where $\omega_0$ is the resonance frequency of the oscillator, $\zeta$
is a damping coefficient, $W$ is a Gaussian white noise and
$t_{\mathrm{end}} = 30\,\text{s}$.
The solution of~\eqref{eq:ex2_stochas_eq} can be approximated using an
exponential Euler scheme with time step $\delta > 0$ (more details can
be found in the Supplementary Material):
we denote by $X_k^{(\delta)}$ the resulting approximation of $X$ at
time steps $t_{k} = k \delta$, $k \in \Nset$,
$k \le K_\delta = \lfloor t_{\mathrm{end}} / \delta \rfloor$.
We will be interested in the maximal log-displacement
$\max_{t \le t_{\mathrm{end}}} \log \left| X(t) \right|$, that we
approximate by
$Z(\omega_0, \zeta, \delta) = \max_{k \le K_\delta} \log \bigl(
|X^{(\delta)}_k| \bigr)$.

We view the mapping
$x=(\omega_0, \zeta, \delta) \in \Rset_+^{3} \mapsto Z(\omega_0,
\zeta, \delta)$ as a multi-fidelity stochastic simulator,
where~$\delta$ controls the level of fidelity.
In this problem, the QoI is the function
$Q: (\omega_0, \zeta) \mapsto \prob(Z(\omega_0, \zeta, \deltaref) >
\zcrit)$, where $\deltaref = 0.01~\text{s}$ denotes the level of
highest fidelity
and $\zcrit = -3$ is a given critical threshold.
The computational cost of $Z$ is an affine function of $1/\delta$:
$C(\delta) = a/\delta + b$.
After normalization to have $C(\deltaref) = 1$, the coefficients are
$a = 0.0098$ and $b = 0.0208$.

\begin{table}
  \begin{center}
    \begin{tabular}{l|*{10}{c}}
      Level $\delta$ & $1\,\text{s}$ & $0.51\,\text{s}$ & $1/3\,\text{s}$
      & $0.25\,\text{s}$ & $0.2\,\text{s}$
      & $1/6\,\text{s}$ & $0.1\,\text{s}$
      & $0.05\,\text{s}$ & $0.02\,\text{s}$ & $0.01\,\text{s}$\\
      \hline
      Cost$^{-1}$ & 32.7 & 24.8 & 19.9 & 16.7
                         & 14.3 & 12.6 & 8.4 & 4.6 & 2 & 1\\
      \hline 
      Initial DoE & 180 & 60 & 20 & 10 & 5 & 0 & 0 & 0 & 0 & 0\\
    \end{tabular}
  \end{center}
  \caption{%
    Levels of fidelity considered in this example.
    The highest level of fidelity is $\delta = 0.01\,\text{s}$.
  }
  \label{tab:ex2_levels}
\end{table}

A good approximation of the output distributions is obtained if we
assume
$Z(\omega_0,\zeta,\delta) \mid \xi, \lambda \sim \N(\xi(x),
\lambda(\delta))$, where the variance only depends on the fidelity
level.
This assumption makes it possible to write
\begin{equation*}
  Q(\omega_0, \zeta) =  \Phi\left(
    \frac{\xi(x) - \zcrit}{\sqrt{\lambda(\deltaref)}}
  \right).
\end{equation*}
The mean function~$\xi$ is modeled by the additive Gaussian process
model~\EquModelTWY of Section~\ref{subsec:model:TWY}, where the
variance~$\lambda$ is log-Gaussian as in
Section~\ref{subsec:model:stoch}
and the prior distributions for the hyper-parameters are set as in
\citet{stroh2017integrating}.
The posterior mean $\widehat Q_n = \esp_{n} (Q)$ is used to
estimate~$Q$.

In this example we consider $S = 10$ levels,
and the initial design is an NLHS on the five first levels.
The different levels of fidelity, their costs, and the initial design
are summarized in Table~\ref{tab:ex2_levels}.
The total cost of the initial design is $9.88$.
The total simulation budget, taking into account the initial budget,
is set to $20$.
We also use a very high simulation budget to compute a reference value
for $Q$, which will be used to assess the estimation error.

\begin{figure}
  \begin{center}
    \psfrag{Cost}[tc][tc]{$c_n$}
    \psfrag{||p(x, tHF) - p*(x, tHF)||}[bc][bc]{$\MedErr_n$}
    \psfrag{Error-L2 on the probability function}[tc][tc]{}
    \includegraphics[width = 0.6\textwidth]{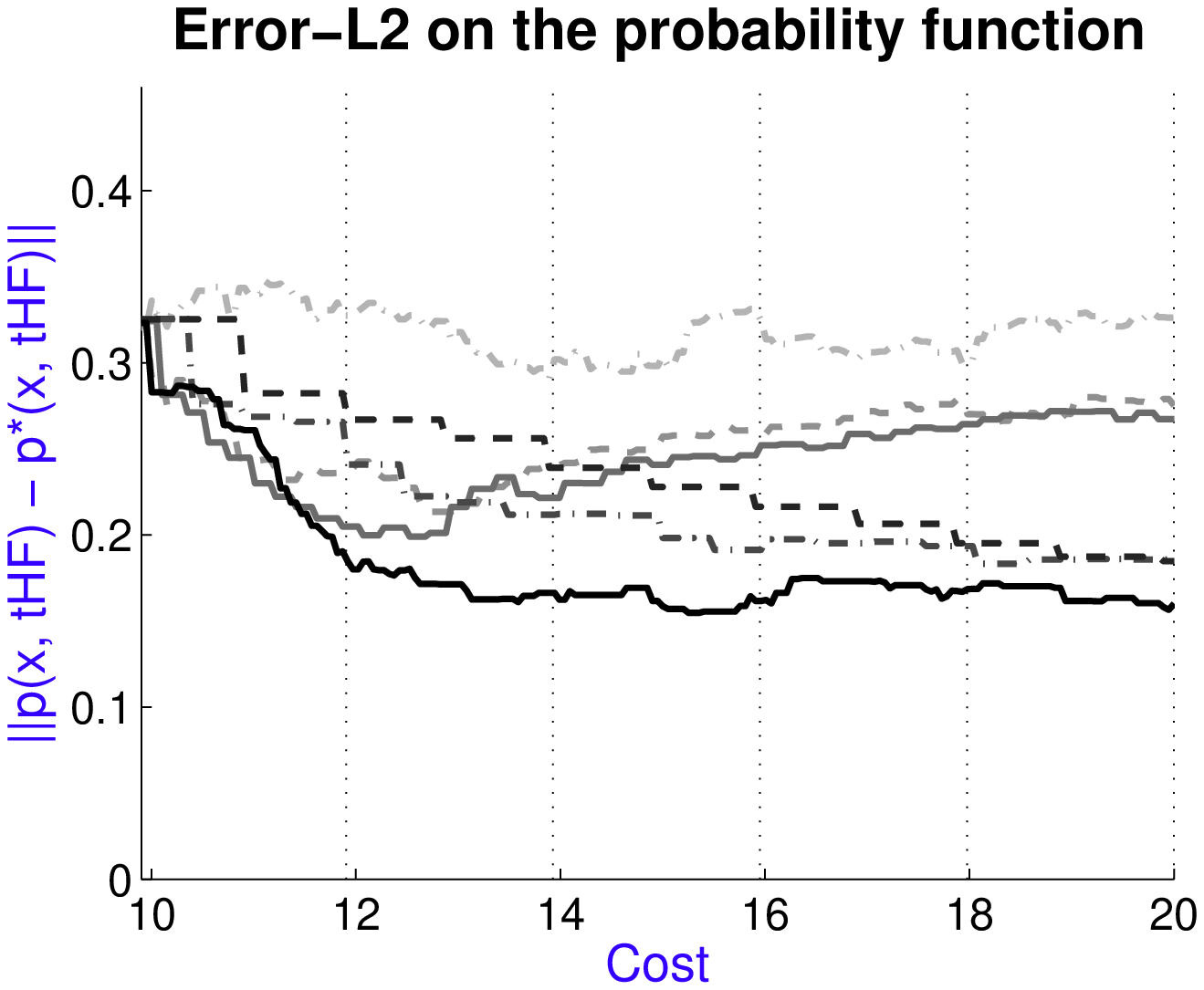}
    \psfrag{SUR|dt=0.167}[cl][cl]{\small SUR at $\delta = 0.16\,\text{s}$}
    \psfrag{SUR|dt=0.1}[cl][cl]{\small SUR at $\delta = 0.1\,\text{s}$}
    \psfrag{SUR|dt=0.05}[cl][cl]{\small SUR at $\delta = 0.05\,\text{s}$}
    \psfrag{SUR|dt=0.02}[cl][cl]{\small SUR at $\delta = 0.02\,\text{s}$}
    \psfrag{SUR|dt=0.01}[cl][cl]{\small SUR at $\delta = 0.01\,\text{s}$}
    \psfrag{MSUR}[cl][cl]{\small MR-SUR}
    \subfloat{
      \includegraphics[width = 0.38\textwidth]{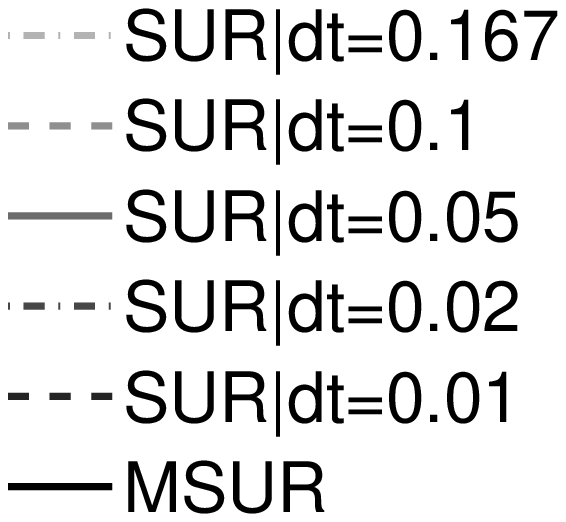}}
    \addtocounter{subfigure}{-1}
  \end{center}
  \caption{%
    Median estimation error as a function of the cost for
    the oscillator test case.
  }
  \label{fig:ex2}
\end{figure}

We compare the MR-SUR strategy, using the integrated posterior
variance~\eqref{equ:Hn-L2-U} as the uncertainty measure, to five
different SUR strategies based on the same uncertainty measure.
All of them are started with the same initial design, and each SUR
strategy corresponds to sampling on only one of the five
highest-fidelity levels.
The experiment is repeated 48 times with different initial designs,
and the strategies are compared using the median estimation error as
in Section~\ref{sec:one-dim-example}.
The result is shown on Figure~\ref{fig:ex2}:
the MR-SUR strategy is never far from the best strategy for any given
budget, and actually outperforms all the other (fixed-level SUR)
strategies as soon as~$c_n$ larger that approximately~11.5.
Additional experiments, presented in the Supplementary Material, also
show the benefit of using MR-SUR with batches of parallel evaluations
on this example.

\subsection{A fire safety example}
\label{subsec:illustration_fds}

In this section, we illustrate the MR-SUR strategy on a fire safety
application.
The goal is to assess the safety of a
$20\,\text{m} \times 12\,\text{m} \times 16\,\text{m}$
parallelepiped-shaped storage facility,
with two $2\,\text{m} \times 1\,\text{m}$ open doors and two
$2\,\text{m} \times 2\,\text{m}$ open windows.
The propagation of smoke and heat is simulated using Fire Dynamics
Simulator \citep[FDS; see][]{mcgrattan2010fire}, a state-of-the-art
CFD software for fire engineering, which solves the transport
equations using finite difference methods.
The fire is located at the center of the room, and burns polyurethane.
To assess fire safety, the values of several physical quantities are
compared against regulatory thresholds---in this illustration, we
focus on one of them only, called visibility and hereafter denoted
by~$V$.
According to the \citet{iso2012life}, visibility must remain greater
than $\zcrit = 5~\text{m}$ to ensure safety during an evacuation.

\begin{table}
  \begin{center}
    \begin{tabular}{l|*{4}{c}}
      Level~$\delta$
      & $50~\text{cm}$
      & $33~\text{cm}$
      & $25~\text{cm}$
      & $20~\text{cm}$
      \\
      \hline
      Real cost
      & $69~\text{min}$
      & $6~\text{h}$
      & $20~\text{h}$
      & $49~\text{h}$
      \\
      Normalized cost
      & 1/42
      & 1/8
      & 1/2.5
      & 1
      \\
      \hline 
      Initial design
      & 90
      & 30
      & 10
      & 0
      \\
    \end{tabular}
  \end{center}
  \caption{The levels of fidelity on FDS.}
  \label{tab:fds_levels}
\end{table}

Our FDS-based simulator will be treated as a stochastic simulator%
\footnote{%
  although is is actually, strictly speaking, a deterministic
  simulator, since the seed of the random number generator is fixed by
  the software; cf. \cite{stroh2017assessing} for details.}
with nine input variables:
three environmental variables (external temperature~$T_{\rm ext}$,
atmospheric pressure~$P_{\rm atm}$ and ambient
temperature~$T_{\rm amb}$) denoted by $u_{\rm e}\in\Rset^{3}$,
five ``scenario variables'' (fire growth rate~$\alpha$, fire
area~$A_{\rm f}$, maximal heat release rate~$\dot{Q}_{\rm h}''$, total
released energy~$q_{\rm f_d}$ and soot yield~$Y_{\rm soot}$) denoted
by $u_{\rm s}\in\Rset^5$,
and finally the size~$\delta$ of the spatial discretization mesh,
which plays the role of a fidelity parameter.
The reader if referred to \cite{stroh2017assessing}
and~\citet{stroh:thesis} for more details on the application.

In this example, our objective is to estimate the probability that $V$
becomes less than $z^{\rm crit}$ in a particular fire scenario,
defined by $\alpha = 0.1057\,\text{kW}\cdot\text{s}^{-2}$,
$A_{\rm f} = 14\,\text{m}^2$,
$\dot{Q}_{\rm h}'' = 460\,\text{kW}\cdot\text{m}^{-2}$,
$q_{\rm f_q} = 450\,\text{MJ}\cdot\text{m}^{-2}$, and
$Y_{\rm soot} = 0.027\,\text{kg}\cdot\text{kg}^{-1}$.
The environmental inputs~$u_{\rm e}$ are assumed random and
integrated according to an environmental distribution~$\Fenv$,
which is a trivariate normal distribution with mean
$(10\,^{\circ}\text{C}, 100\,\text{kPa}, 22.5\,^{\circ}\text{C})$,
variances equal to
$(20/3\,^{\circ}\text{C},2/3\,\text{kPa}, 2.5\,^{\circ}\text{C})^2$,
and a correlation coefficient of 0.8 between the temperatures.
The QoI is
$Q = \int_{\Uset_{\rm e}} p(u_{\rm e}, u_{\rm s}, \deltaref)\, \ddiff
\Fenv(u_{\rm e})$,
where $\deltaref = 20\,\text{cm}$ is the reference level and
$p(u_{\rm e}, u_{\rm s}, \deltaref) = \prob\left(V < \zcrit\vert
  u_{\rm e}, u_{\rm s}, \deltaref\right)$.

Four levels of fidelity will be considered for running simulations:
the reference level $\deltaref = 20\,\text{cm}$,
and three levels of lower fidelity ($\delta = 50\,\text{cm}$,
$33\,\text{cm}$ and $25\,\text{cm}$).
Table~\ref{tab:fds_levels} shows the correspondence between levels and
computation times.
Four independent initial NLHS designs of size~$n = 130$, distributed
across the first three levels of fidelity as shown in
Table~\ref{tab:fds_levels}, are available from previous studies.
The normalized cost of each initial design is 9.89 (i.e., 20 days).
A reference value for $Q$ has also been obtained from 150 Monte Carlo
simulations, distributed on the highest fidelity level using $\Fenv$.
This reference value has a normalized cost of 150 (i.e., about ten
months).

We run the MR-SUR strategy starting from our four initial designs,
using a supplementary budget of 24 for each run (about 48.7 days).
The underlying Bayesian model is the same as in
Section~\ref{subsec:illustration_dampedOscillator},
the QoI $Q$ is estimated using the posterior mean
$\hat Q_n = \esp_{n}(Q)$,
and the measure of uncertainty is the integrated posterior variance
\begin{equation*}
  H_n = \int \var_n \left( %
    p(u_{\rm e}, u_{\rm s}, \deltaref) %
  \right)\, \ddiff \Fenv(u_{\rm e}),
\end{equation*}
which is a special case of~\eqref{equ:Hn-L2}.
The corresponding SUR criterion is similar
to~\eqref{eq:notre-SUR-crit}, with an integral over the environmental
variables only.
The result is shown on Figure~\ref{fig:ex3_res}.
We can see on Figure~\ref{fig:ex3_res_estimVI} that the estimations
are initially, in three out of four cases, incompatible with the Monte
Carlo one, but tend to get closer to the reference value when more
simulations are carried out using the MR-SUR strategy.
Figure~\ref{fig:ex3_res_uncerVI} shows the measure of uncertainty as a
function of the cost:
the uncertainty is large at the beginning of the sequential design and
rapidly becomes smaller, as expected, as the MR-SUR strategy proceeds.
(Note that the cost of the whole design is approximately
$9.9 + 24 = 33.9$, which is must cheaper than the cost of~150 of the
Monte Carlo reference.)

\begin{figure}
  \begin{center}
    \psfrag{Cost}[tc][tc]{$c_n$}
    \psfrag{Phat}[bc][bc]{$\widehat Q_n$}
    \psfrag{sqrt(Hn)}[bc][bc]{$\sqrt{H_n}$}
    \psfrag{Estimations of the probability}[bc][bc]{}
    \psfrag{Measure of uncertainty}[bc][bc]{}

    \subfloat[Estimations of the probability]{
      \includegraphics[width=0.45\textwidth]{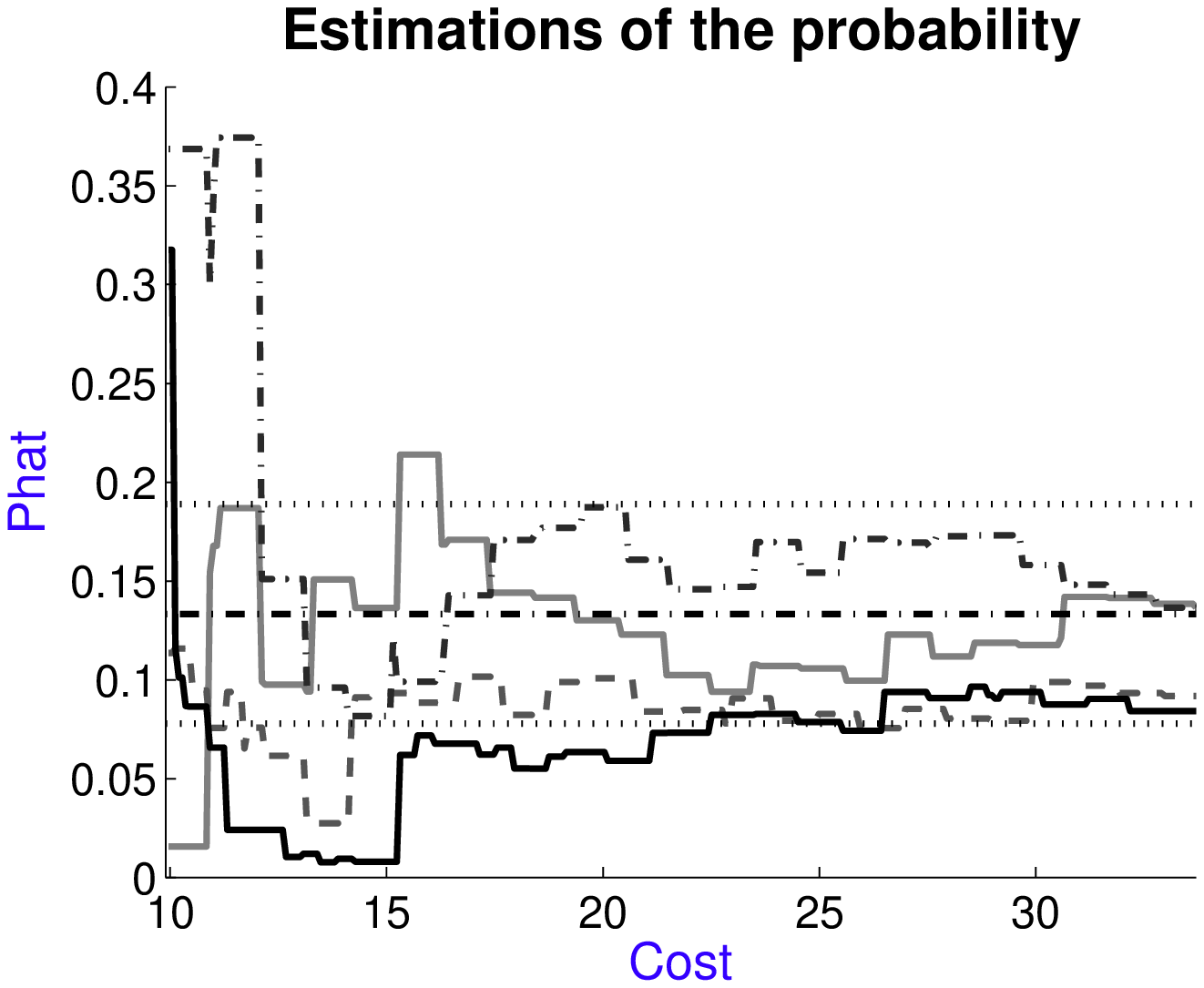}
      \label{fig:ex3_res_estimVI}}
    \subfloat[Measures of uncertainty]{%
      \includegraphics[width=0.45\textwidth]{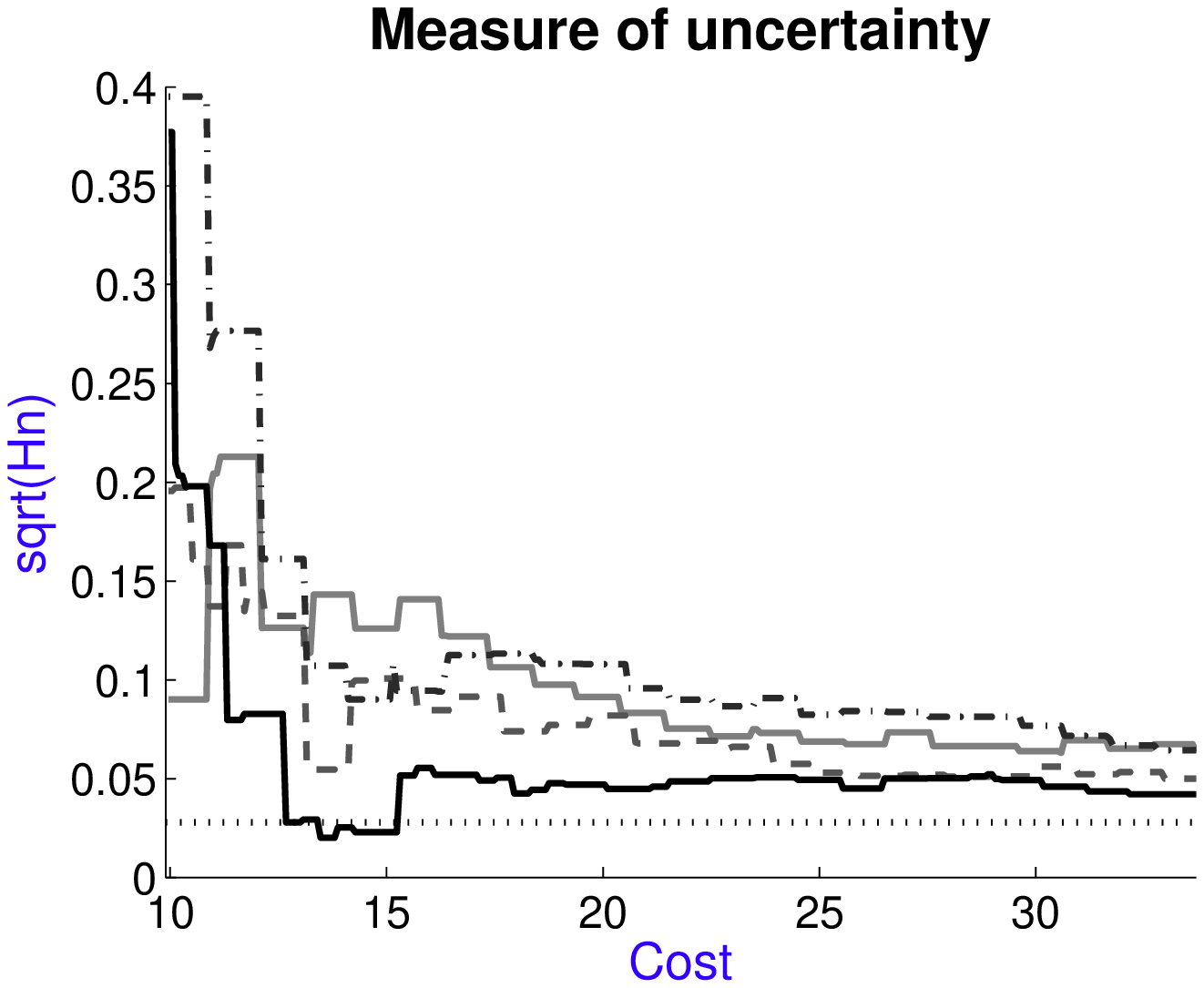}
      \label{fig:ex3_res_uncerVI}%
    }
  \end{center}
  \caption{%
    Result of four repetitions for the fire safety example. %
    (a)~Estimated probability as a function of the cost. %
    The horizontal lines correspond to the Monte Carlo reference
    (dash-dotted line: mean; dotted lines: two-standard-deviation
    interval).
    (b)~Square root of the measure of uncertainty
    (upper bound on the posterior standard deviation of the probability).
    The horizontal dotted line is the Monte Carlo standard deviation.
  }
  \label{fig:ex3_res}
\end{figure}

\section{Conclusion}
\label{sec:conclusion}

The main contribution of this article is to unify and extend several methods of
the literature of Bayesian sequential design of experiments for
multi-fidelity numerical simulators.
The unification that we propose is cast in the framework of Stepwise
Uncertainty Reduction (SUR) strategies:
when the accuracy of computer simulations can be chosen by the user, a
natural extension of SUR strategies is to consider sampling criteria
built as the ratio between the reduction of uncertainty and the cost
of a simulation.
We call this approach Maximal Rate of Stepwise Uncertainty Reduction
(MR-SUR).
It can be applied to deterministic or stochastic simulators.
Our numerical experiments show that the MR-SUR approach typically
provides estimations which, for a given computational cost, are never
much worse, and often better, than the best SUR strategy using a
single level of fidelity.

Further work directions could be considered in the future.
For instance, there is no explicit ingredient in MR-SUR strategies
that tells the procedure to ``learn'' the model, and in particular, to
learn the correlations between the levels of fidelity.
It seems to us that this would be important, particularly when
simulations are very expensive and the simulation budget is very
limited, as in our fire safety application.
Using a fully Bayesian approach would somehow answer this problem, as
the uncertainty about the model would be propagated to the uncertainty
about the QoI.

Another important research direction would be to address parallel
simulations.
What would be a principled approach of resource allocation when
several simulations with different accuracies and different costs can
be conducted at the same time?

\bibliographystyle{apalike}
\bibliography{mrsur}

\newpage \appendix  \let\citet\citeSM

\setcounter{equation}{0}  \renewcommand\theequation {SM\arabic{equation}}
\setcounter{figure}{0}    \renewcommand\thefigure   {SM\arabic{figure}}
\setcounter{table}{0}     \renewcommand\thetable    {SM\arabic{table}}
\setcounter{section}{0}   \renewcommand\thesection  {SM-\arabic{section}}
\setcounter{remark}{0}    \renewcommand\theremark   {SM\arabic{remark}}

\begin{center}
  \large\bf SUPPLEMENTARY MATERIAL
\end{center}

\addtocontents{toc}{\protect\vspace{30pt}}

\section{Introduction}

This document contains supplementary material for the article %
``Sequential design of multi-fidelity computer experiments: %
maximizing the rate of stepwise uncertainty reduction''.
It is organized as follows.
Section~\ref{append:initial_design} provides a short literature review
on non-sequential designs for multi-fidelity, which complements the
literature review on sequential designs given in Section~3.1 of the
article.
Section~\ref{append:surCrit} provides (with proof) a new SUR
criterion, which covers as a special case the criterion provided
(without proof) in Section~3.2 of the article.
Finally, Sections~\ref{sec:SM:1d-example} and~\ref{sec:SM:damped}
provide additional information regarding the examples presented in
Sections~4.2 and~4.3 of the article.

\section{Non-sequential designs in multi-fidelity}
\label{append:initial_design}

In this section, we provide a very brief literature review on
non-sequential designs for multi-fidelity.

A common recommendation for multi-level designs is nesting. %
A multi-level design is nested when any observed point at a
level~$\delta^{(s)}$ is also observed at every lower-fidelity
level~$\delta^{(s')}$, $s' < s$. %
Furthermore, space-filling designs are also expected to ensure
observations in the whole input domain, as usual in Gaussian process
regression.

A simple method to create a nested design is proposed by
\citet{forrester2007multi}. %
It draws a maximin Latin Hypercube Sampling (LHS) at the lowest
fidelity design, and then selects subsets of this LHS at the next
levels. %
\citet{le2014recursive} suggest to start with the highest-fidelity
level and add points on the lower-fidelity levels to ensure better
space-filling properties. %
The reader is also referred to \citet{rennen2010nested} for a method
which applies when there is only two levels of fidelity.

In our work, we use the method proposed by \citet{qian2009nested},
which construct Nested Latin Hypercube Sampling (NLHS). %
An NLHS is a nested design with the property of being an LHS at each
level of fidelity. %
We add a maximin optimization step to ensure that the design is
space-filling at each level.

Note that this method was extended in several directions. %
\citet{he2011nested}, \citet{yang2014construction},
\citet{guo2017construction}, and \citet{xu2017general} propose methods
to build NLHS with particular structures, such as orthogonality.

\section{A new SUR criterion}
\label{append:surCrit}

\subsection{Criterion definition and result statement}

\newcommand \Pe {\prob_e}

Let $\xi = \left( \xi(x) \right)_{x \in \Xset}$ denote a Gaussian
process prior for the mean response of a stochastic simulator with
Gaussian responses---i.e., a simulator which produces random responses
\begin{equation*}
  Z \mid \xi \;\sim\; \mathcal{N} \left( \xi(x), \lambda(x) \right),
\end{equation*}
where $x \in \Xset$ denotes the vector of inputs of the simulator and
$\lambda: \Xset \to \left[ 0, +\infty \right)$ is a known variance
function.
Assume that the quantity of interest is the probability function
$\alpha: \Xset \to \left[ 0, 1 \right]$ defined by
\begin{equation*}
  \alpha(x) \;=\; \prob\left( Z_x > \zcrit \right),
\end{equation*}
where $\zcrit \in \Rset$ is a given threshold and $Z_x$ denotes a
(future) response of the simulator with $x$ as the input.
More explicitely, we have
\begin{equation}
  \label{equ:alpha-explicit}
  \alpha(x) \;=\; \Phi\left(
    \frac{\xi(x) - \zcrit}{\sqrt{\lambda(x)}}
  \right).
\end{equation}

Let $\widehat\alpha_n(x)$ denote the posterior mean of~$\alpha(x)$
given $n$ responses~$Z_1$, \ldots, $Z_n$ of the simulator at design
points~$X_1, \ldots, X_n \in \Xset$ (possibly selected in a sequential
manner).
Let $\mu$ denote a positive, bounded measure on~$\Xset$,
and consider the measure of uncertainty~$H_n$ defined by
\begin{equation*}
  H_n
  = \esp_n \left(
    \int \left( \alpha(x) - \widehat\alpha(x) \right)^2
    \mu(\ddiff x)
  \right)
  = \int
    \var_n \left( \alpha(x) \right)\, \mu(\ddiff x).
\end{equation*}
The following result provides a tractable expression for the
corresponding SUR criterion for a batch
$\xx = (\xx_l)_{1\leq l\leq q}\in \Xset^q$ of candidate points.
The criterion presented in the article corresponds to the fully
sequential case---i.e., $q = 1$---and to a particular choice of the
measure~$\mu$.

\begin{proposition} \label{prop:sur}
  Let $m_n$ (resp. $k_n$) denote the posterior mean (resp. the
  posterior covariance) of~$\xi$ given $n$ observations.
  Let
  \begin{equation*}
    J_n(\xx) \;=\; \esp\left(
      H_{n+q} \bigm| X_{n+1} = \xx_1,\, \ldots,\, X_{n+q} = \xx_q
    \right)
  \end{equation*}
  and $G_n(\xx) = H_n - J_n(\xx)$.
  Then,
  \begin{equation*}
    J_n(\xx) \;=\; \int \Bigl[ \Phi_2(u_n(x), a_n(x); r_n(x))
    - \Bigl. \Phi_2(a_n(x), a_n(x);
    \widetilde{r}_n(x; \xx)) \Bigr]\,
    \mu(\ddiff x)
  \end{equation*}
  and
  \begin{equation*}
    G_n(\xx) \;=\; \int \Bigl[ \Phi_2(a_n(x), a_n(x);
      \widetilde{r}_n(x; \xx)) \Bigr.
      - \Bigl. \Phi(a_n(x))^2 \Bigr]\, \mu(\ddiff x),
  \end{equation*}
  with $\Phi$ the cdf of the standard normal distribution,
  $\Phi_2 \left( \cdot, \cdot \,; \rho \right)$ the cdf of the
  standard bivariate normal distribution with correlation~$\rho$, and
  \newcommand\tempVP{\vphantom{\Bigm|}}
  \begin{equation*}
    \begin{array}{rcl|rcl}
      \tempVP
      a_n(x)
      & = & \left( m_n(x) - \zcrit \right) / \sqrt{v_n(x)}, \;
      & v_n(x)
      & = & k_n(x, x) + \lambda(x),
      \\
      \tempVP
      r_n(x)
      & = & k_n(x, x) / v_n(x),
      & \; \widetilde{r}_n(x; \xx)
      & = & \nu_n(x, x; \xx) / v_n(x),
      \\
      \tempVP
      \nu_n(x, x';\xx)
      & = &
      \multicolumn{4}{l}{
      k_n(x, x') - k_{n+q}(x, x'; \xx)
            = k_n(\xx, x)^TK_n(\xx, \xx)^{-1}k_n(\xx, x'),}
      \\
      \tempVP
      k_n(\xx, x)
      & = & \left(k_n(\xx_l, x)\right)_{1\leq l\leq q},
      & K_n(\xx, \xx)
      & = & \left(k_n(\xx_l, x_{l'})
        + \lambda(\xx_l)\cdot \delta_{l=l'})\right)_{1\leq l,l'\leq q}.
    \end{array}
  \end{equation*}
\end{proposition}

\bigskip

\begin{remark}
  In the expressions of~$J_n(\xx)$ and~$G_n(\xx)$, the only
  part which depends on the future design $\xx$ is
  $\int \Phi_2(a_n(x), a_n(x); \widetilde{r}_n(x; \xx))\,
  \mu(\ddiff x)$, which must be maximized.
\end{remark}

\subsection{A useful identity}

\newcommand \wtPhi {\widetilde\Phi}

Let $\wtPhi_d\left(\, \cdot \;; m, K \right)$ denote the cumulative
distribution function of the $d$-variate normal distribution with
mean~$m$ and covariance matrix~$K$.
The following identity is used by \citet{chevalier2014fast} for the
computation of SUR criteria similar to ours.

\begin{lemma} \label{lem:trick}
  Let $W \sim \Ncal (m, K)$ be a $d$-dimensional Gaussian vector.
  Then, for any mean vector~$m'$ and covariance matrix~$K'$,
  \begin{equation*}
    \esp\left( \wtPhi_d(W; m', K')\right) = \wtPhi_d(m; m', K + K').
  \end{equation*}
\end{lemma}

\begin{proof}
  Let $W' \sim \mathcal{N}(m', K')$ be independent from~$W$.  Then
  \begin{equation*}
    \esp\left( \wtPhi_d(W, m', K')\right)
    \;=\; \esp\left( \prob\left( W' \leq W \bigm| W \right) \right)
    \;=\; \prob(W' \leq W),
  \end{equation*}
  and, using that $W' - W = W'' - m$ with $W'' \sim \Ncal(m', K + K')$,
  \begin{equation*}
    \prob(W' \leq W)
    = \prob(W'' \leq m)
    = \wtPhi_d(m; m', K + K').   \vspace{-2.6em}
  \end{equation*}
\end{proof}

\begin{corollary} \label{cor:trick}
  Let $m, m' \in \Rset$, $v, v' \in \left[ 0, +\infty \right)$ and
  $W \sim \Ncal (m, v)$.  Then
  \begin{align*}
    \esp\left( \Phi\left( \frac{W - m'}{\sqrt{v'}} \right) \right)
    & \;=\; \Phi\left( \frac{m - m'}{\sqrt{v + v'}} \right)\\
    \intertext{and}
    \esp\left( \Phi\left( \frac{W - m'}{\sqrt{v'}} \right)^2 \right)
    & \;=\; \Phi_2\left( %
      \frac{m - m'}{\sqrt{v + v'}},
      \frac{m - m'}{\sqrt{v + v'}};
      \frac{v}{v + v'}
      \right).
  \end{align*}
\end{corollary}

\subsection{Proof of Proposition~\ref{prop:sur}}

\begin{proof}

  Recall from~\eqref{equ:alpha-explicit} that
  \begin{equation*} \label{equ:alpha-explicit-recall}
    \alpha(x) \;=\; \Phi\left(
      \frac{\xi(x) - \zcrit}{\sqrt{\lambda(x)}}
    \right).
  \end{equation*}
  It follows from Corollary~\ref{cor:trick} that
  \begin{equation} \label{equ:esp-alpha-x}
    \esp_{n+q} \left( \alpha(x) \right)
    \;=\;
    \Phi\left( \frac{m_{n+q}(x) - \zcrit}{%
        \sqrt{\lambda(x) + k_{n+q}(x, x)}} \right)
    \;=\; \Phi\left( \frac{m_{n+q}(x) - \zcrit}{%
        \sqrt{v_{n+q}(x)}} \right),
  \end{equation}
  and thus
  \begin{equation} \label{equ:var-q-puls-q}
    \var_{n+q} \left( \alpha(x) \right)
    \;=\;
    \esp_{n+q} \left( \alpha(x)^2 \right)
    - \Phi\left( \frac{m_{n+q}(x) - \zcrit}{%
        \sqrt{v_{n+q}(x)}} \right)^2.
  \end{equation}

  \medbreak

  Let us now compute separately the expectation with respect
  to~$\prob_n$ of the two terms in the right-hand side
  of~\eqref{equ:var-q-puls-q}.
  For the first term we have
  \begin{align}
    \esp_n \left(  \esp_{n+q} \left( \alpha(x)^2 \right) \right)
    & \;=\; \esp_n \left( \alpha(x)^2 \right)
      \;=\; \esp_n \left( \Phi \left( %
        \frac{\xi(x) - \zcrit}{\sqrt{\lambda(x)}}
      \right)^2 \right) \nonumber\\
    & \;=\; \Phi_2\left( a_n(x), a_n(x); r_n(x) \right),
      \label{equ:first-term}
  \end{align}
  where we have applied the second part of Corollary~\ref{cor:trick}
  with $m = m_n(x)$, $v = k_n(x,x)$, $m' = \zcrit$, and
  $v' = \lambda(x)$.
  For the second term we observe that, under~$\prob_n$, $m_{n+q}$~is a
  Gaussian process with mean~$m_n$ and covariance
  function~$\nu_n( \,\cdot,\, \cdot\, ; \xx )$.
  Therefore
  $m_{n+q}(x) \sim \Ncal \left( m_n(x), \nu(x,x; \xx \right)$,
  and it follows that
  \begin{equation}
    \esp_n \left(
      \Phi\left( \frac{m_{n+q}(x) - \zcrit}{%
          \sqrt{v_{n+q}(x)}} \right)^2
    \right) \;=\; \Phi_2 \left(
      a_n(x), a_n(x); \widetilde r_n(x) \right),
    \label{equ:second-term}
  \end{equation}
  where we have used again the second part of
  Corollary~\ref{cor:trick} with $m = m_n(x)$,
  $v = \nu_n(x,x; \xx)$, $m' = \zcrit$ and
  $v' = v_{n+q}(x)$.
  Indeed,
  \begin{align*}
    v + v' & \;=\; %
             \nu_n(x,x; \xx) + v_{n+q}(x)\\
    & \;=\; %
    \left( k_n(x, x) - k_{n+q}(x,x) \right)
    + \left( k_{n+q}(x,x) + \lambda(x) \right)
    \;=\; v_n(x),
  \end{align*}
  therefore
  \begin{align*}
     \frac{m - m'}{\sqrt{v + v'}}
     & \;=\; \frac{ m_n(x) - \zcrit}{\sqrt{v_n}}
       \;=\; a_n(x),\\
    \frac{v}{v + v'}
     & \;=\; \frac{ \nu_n(x,x; \xx)}{v_n(x)}
       \;=\; \widetilde r_n(x).
  \end{align*}
  Combining \eqref{equ:var-q-puls-q}--\eqref{equ:second-term} and
  integrating on~$\Xset$ with respect to~$\mu$ yields the desired
  expression for~$J_n(x)$.
  Similarly, combining~\eqref{equ:esp-alpha-x} with $q = 0$
  and~\eqref{equ:first-term} we have
  \begin{equation*}
    H_n \;=\; \int \left(
      \Phi_2\left( a_n(x), a_n(x); r_n(x) \right)
      - \Phi\left( a_n(x) \right)^2
    \right)\, \mu(\ddiff x)
  \end{equation*}
  and the expression of~$G_n(\xx)$ follows.
\end{proof}

\section{One-dimensional example}
\label{sec:SM:1d-example}

This section provides additional information regarding the Bayesian
model used in the ``One-dimensional example'' presented in Section~4.2
of the main article.

The model used in this example is the one proposed by
\citet{kennedy2000predicting} and reviewed in Section~2.1 of the
article, with $S = 2$ levels.
The two independent Gaussian processes~$\eta_1$ and~$\eta_2$ are
stationary processes with unknown constant means and Matérn
covariance functions with regularity 5/2:
\begin{equation*}
  \eta_s \sim \GP \left( %
    m_s,\, \sigma^2_s\matern_{5/2}\left(a_s(\cdot - \cdot)\right)\right),
  \qquad s \in \{ 1, 2 \}.
\end{equation*}
Independent priors are used for all the remaining hyper-parameters of
the model:
\begin{itemize}
\item Improper uniform prior distributions on~$\Rset$ are used for the
  means~$m_s$.
\item The parameters of the covariance function follow log-normal
  distributions:
  \begin{align*}
    \log(\sigma_s^2) & \;\sim\; \Ncal(2\log(0.2), \log(100)^2),\\
    \log(a_s) & \;\sim\; \Ncal(\log(2), \log(10)^2).
  \end{align*}
\item Finally, the regression term between the two levels follows a
  normal prior distribution:
  \begin{equation*}
    \rho \;\sim\; \Ncal(1, 2^2).  
  \end{equation*}
\end{itemize}

\section{Random damped harmonic oscillator}
\label{sec:SM:damped}

This section provides additional information regarding the ``Random
damped harmonic oscillator'' example (Section~4.3 of the main
article).

\subsection{Explicit exponential Euler scheme}

Consider a stochastic equation
\[
\D{X_t} = AX_t\D{t} + b\D{W_t},
\]
where $X$ is a stochastic vector, $W$ is a Wiener process,
$b$ a real matrix and $A$ a matrix.
Let $\delta$ a time step, we would like to approximate
$X$ by a finite difference method:
$\widehat{X^{(\delta)}_n}\approx X(\delta n)$.
The explicit exponential Euler scheme is a finite
difference method proposed by \citet{jentzen2009overcoming}
to ensure stability when approximating a stochastic equation.
The method is to apply recursively the formula
\[
\widehat{X^{(\delta)}_{n+1}} = \exp\left(A\delta\right)
\left[\widehat{X^{(\delta)}_{n}} + \sqrt{2\pi S\delta}\cdot b \cdot U\right]
\]
with $S$ the spectral intensity of the Brownian motion,
and $U$ a normal random vector.

In particular, for the application
of the section~4.2,
(14) can be rewritten
\[
\D{\begin{pmatrix}
X_t\\\dot{X_t}\\
\end{pmatrix}} =
\begin{pmatrix}
0 & 1\\-\omega_0^2 & -2\zeta\omega_0\\
\end{pmatrix}
\begin{pmatrix}
X_t\\\dot{X_t}\\
\end{pmatrix}\D{t}
+ \begin{pmatrix}
0\\1\\
\end{pmatrix}\D{W_t}.
\]
Consequently, the approximation with a finite time-step $\delta$
is
\[
\begin{pmatrix}
\widehat{X^{(\delta)}_{n+1}} \\
\widehat{V^{(\delta)}_{n+1}} \\
\end{pmatrix} =
\exp\begin{pmatrix}
0 & \delta\\
-\omega_0^2\delta & -2\zeta\omega_0\delta\\
\end{pmatrix}
\left[
\begin{pmatrix}
\widehat{X^{(\delta)}_{n}} \\
\widehat{V^{(\delta)}_{n}} \\
\end{pmatrix}
+ \begin{pmatrix}
0\\ \sqrt{2\pi S\delta}u\\
\end{pmatrix}\right], \quad u\sim \Ncal(0, 1)
\]

\subsection{Supplementary experiment: batches of parallel evaluations}

In this section, we consider a batch-sequential version of the MR-SUR
strategy, where observations are taken in batches of $q\geq 1$ simulations
having the same computational cost.  The sampling criterion for
parallel synchronous evaluations can be written as
\begin{equation}
({u^{(1)}}^\star, \dots, {u^{(q)}}^\star; \delta^\star)
= \argmax_{\left(u^{(l)}\right)_{1\leq l\leq q} \in \Uset^q,
\delta \in \Rset^+}
\frac{H_n - \esp_n\left[H_{n+q}\vert
\underline{X} = \left((u^{(1)}, \delta),\dots, (u^{(q)}, \delta)\right)
\right]}{C(\delta)}.
\end{equation}

We test this parallel version of MR-SUR on the example of the random
damped oscillator.  We compare four sequential DoE: a SUR strategy on
the highest-fidelity level with one new observation at each iteration
($q = 1$); a parallel SUR strategy on the highest-fidelity level with
$q = 5$; an MR-SUR strategy ($q = 1$); and the parallel version of
MR-SUR with $q = 5$. Each experiment is repeated 24 times.

The comparison between the four sequential designs is shown in
Figure~\ref{fig:ex2_batch}, which represents the $\mathbb{L}^2$
estimation error as a function of the (clock-wall) cost of observations.
We can see that the MR-SUR with $q = 5$ is much more efficient than the
MR-SUR with $q = 1$.  Figure~\ref{fig:ex2_batch_cumul} shows the error
as a function of the cumulative cost of all experiments in parallel.
Notice that the parallel version of MR-SUR provides better results than
the SUR strategies. The parallel version of the MR-SUR strategy is
almost as good as the single-evaluation version at the end of the
procedure. These results suggest that the parallel version of the MR-SUR
strategy should be used when it is possible.

\begin{figure}
\begin{center}
\psfrag{Cost}[tc][tc]{Real cost}
\psfrag{||p(x, tHF) - p*(x, tHF)||}[bc][bc]{
$\left\|Q^{\star} - Q_n\right\|$}
\psfrag{Error-L2 on the probability function}[tc][tc]{}
\subfloat[Results as function of the real cost]{
\includegraphics[width = 0.45\textwidth]{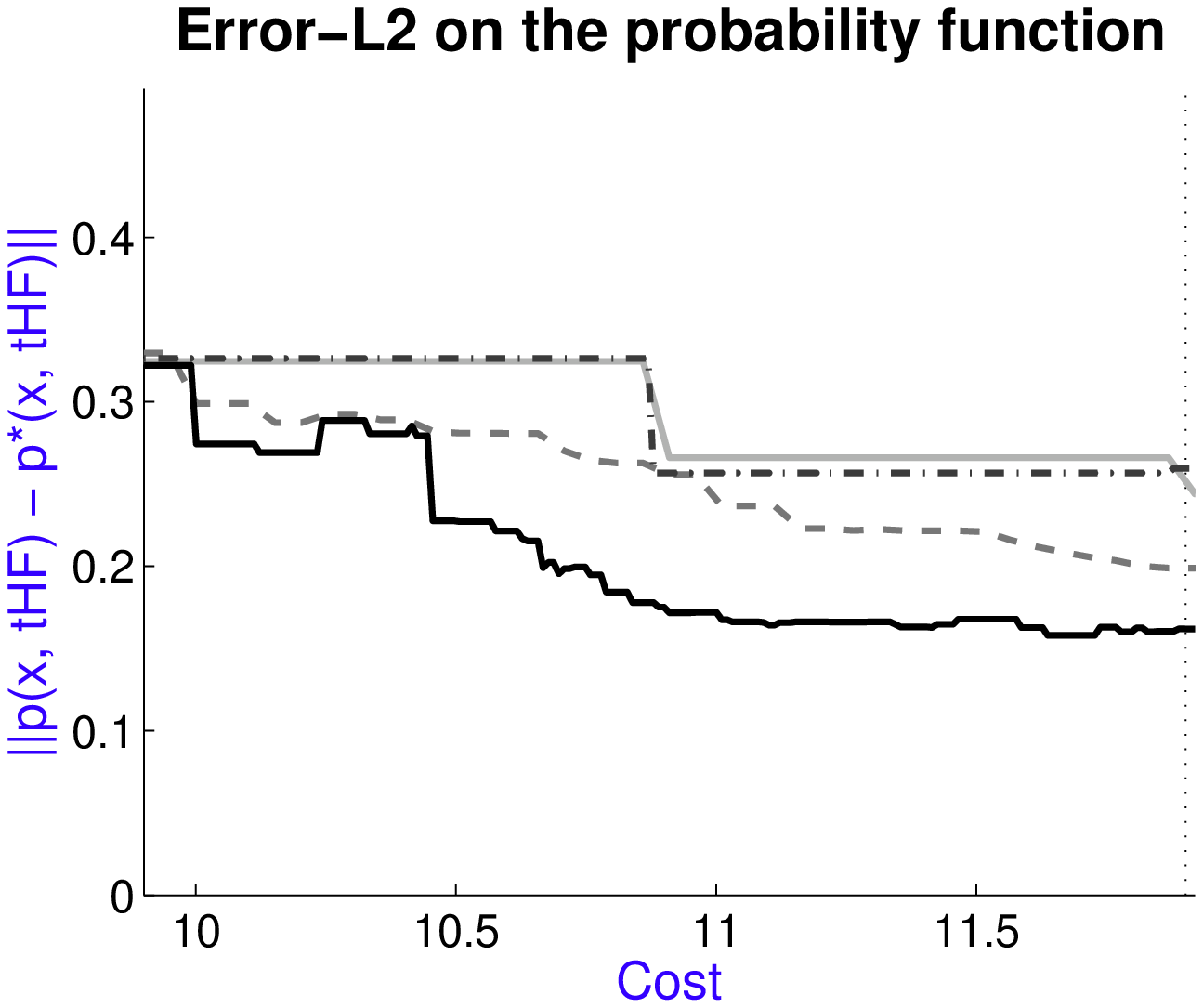}
\label{fig:ex2_batch_real}}
\psfrag{Cost}[tc][tc]{Sum of costs of experiments}
\psfrag{Error-L2 on the probability function}[tc][tc]{}
\subfloat[Results as function of the total cost]{
\includegraphics[width = 0.45\textwidth]{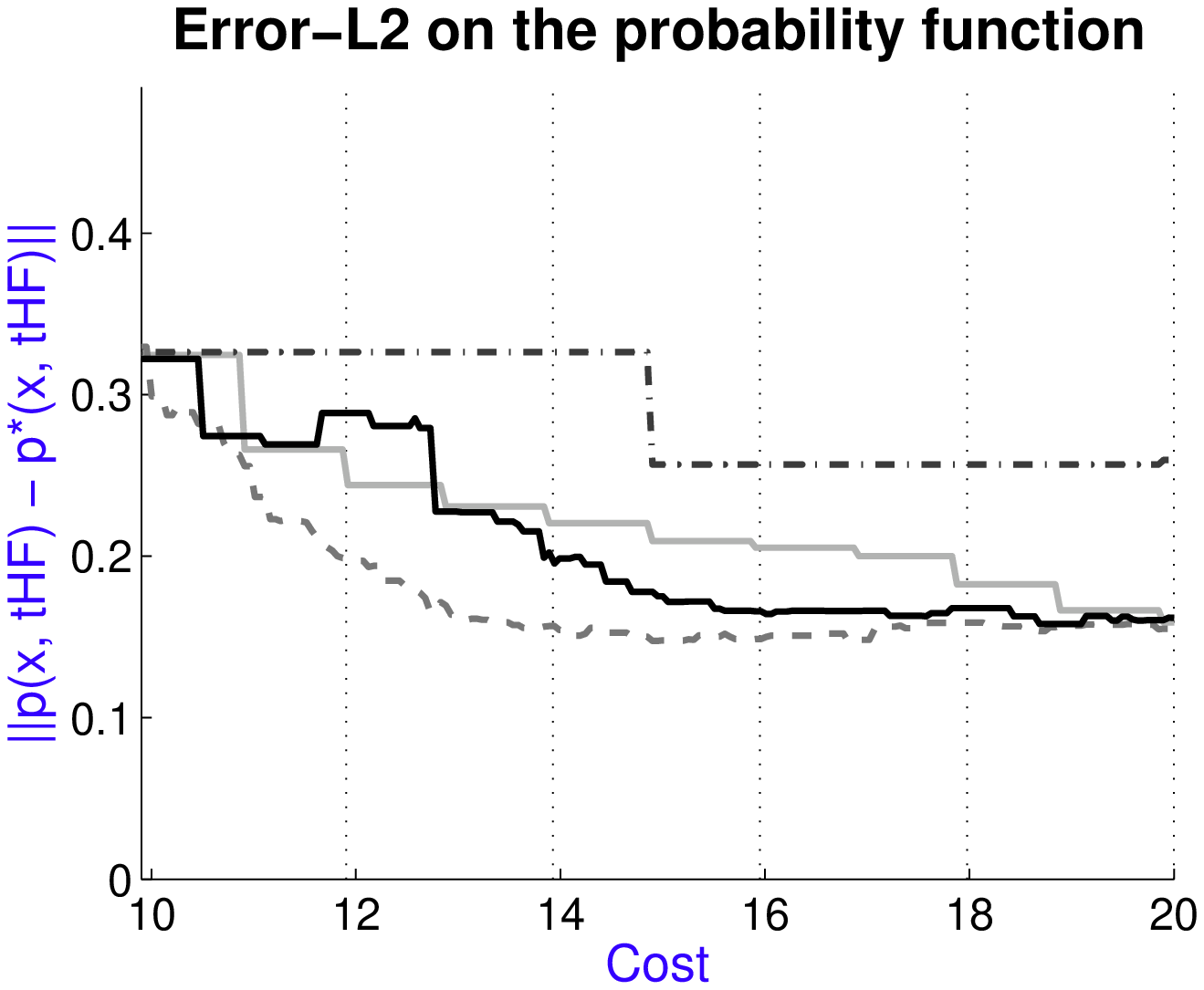}
\label{fig:ex2_batch_cumul}}
\end{center}
\caption{Error between the estimation and the reference value as a
  function of the cost.  Each curve corresponds to one strategy. Light
  solid line: SUR with $q = 1$; gray dashed line: MR-SUR with $q = 1$; dark
  dotted line: SUR with $q = 5$;  black solid line: MR-SUR with $q = 5$.
  (The curves are the median on the 24 repetitions.)}
\label{fig:ex2_batch}
\end{figure}

\bibliographystyleSM{apalike}
\bibliographySM{mrsur}

\end{document}